\documentclass[11pt,a4paper,oneside]{article}
\pdfoutput=1
\usepackage[ascii]{inputenc}

\usepackage[margin=2.5cm]{geometry} %, includefoot, footskip=30pt]{geometry}
\usepackage[dvipsnames]{xcolor}
\usepackage[bookmarksnumbered,linktocpage,hypertexnames=false,colorlinks=true,linkcolor=NavyBlue,urlcolor=NavyBlue,citecolor=ForestGreen,anchorcolor=green,breaklinks=true,pagebackref=true,pdfusetitle]{hyperref}

\usepackage{microtype}
\usepackage[T1]{fontenc}
\usepackage{textcomp}
\usepackage[english]{babel}
\usepackage{textcomp}

\usepackage{amsmath}
\usepackage{amssymb}
\usepackage{amsfonts}
\usepackage{amsthm}
\usepackage{mathtools}
\usepackage{bbm}
\usepackage{bm}
\usepackage{braket}
\usepackage{mathrsfs}
\usepackage{latexsym}
\usepackage{mathdots}
\usepackage{mleftright}

\usepackage{graphicx}
\usepackage{ytableau}
\ytableausetup{boxsize=.8em}
\usepackage{booktabs}
\usepackage{caption}
\usepackage{subcaption}
\usepackage{float}

\usepackage[capitalize, nameinlink]{cleveref}

\usepackage{enumitem}
\usepackage{xspace}
\usepackage{tikz}
\usepackage{quantikz}
\usepackage{authblk}

\theoremstyle{plain}
\newtheorem{thm}{Theorem}[section]\crefname{thm}{Theorem}{Theorems}
\newtheorem{lem}[thm]{Lemma}\AddToHook{env/lem/begin}{\crefalias{thm}{lem}}\crefname{lem}{Lemma}{Lemmas}
\newtheorem{prp}[thm]{Proposition}\AddToHook{env/prp/begin}{\crefalias{thm}{prp}}\crefname{prp}{Proposition}{Propositions}
\newtheorem{cor}[thm]{Corollary}\AddToHook{env/cor/begin}{\crefalias{thm}{cor}}\crefname{cor}{Corollary}{Corollaries}
\AddToHook{env/cnj/begin}{\crefalias{thm}{cnj}}\crefname{cnj}{Conjecture}{Conjectures}
\theoremstyle{definition}
\newtheorem{dfn}[thm]{Definition}\AddToHook{env/dfn/begin}{\crefalias{thm}{dfn}}\crefname{dfn}{Definition}{Definitions}
\newtheorem{prob}[thm]{Problem}\AddToHook{env/prob/begin}{\crefalias{thm}{prob}}\crefname{prob}{Problem}{Problems}
\newtheorem{rem}[thm]{Remark}\AddToHook{env/rem/begin}{\crefalias{thm}{rem}}\crefname{rem}{Remark}{Remarks}
\newtheorem{exa}[thm]{Example}\AddToHook{env/exa/begin}{\crefalias{thm}{exa}}\crefname{exa}{Example}{Examples}

\numberwithin{equation}{section}
\DeclareMathOperator{\QFT}{QFT}

\DeclareMathOperator{\poly}{poly}
\DeclareMathOperator{\polylog}{polylog}

\DeclareMathOperator{\GL}{GL}

\DeclareMathOperator{\U}{U}

\DeclareMathOperator{\Alt}{Alt}
\DeclareMathOperator{\Sym}{Sym}

\DeclarePairedDelimiter\abs{\lvert}{\rvert}

\DeclarePairedDelimiter\parens{\lparen}{\rparen}

\newcommand{\N}{\mathbbm{N}}

\newcommand{\C}{\mathbbm{C}}

\newcommand{\NP}{\mathsf{NP}}
\newcommand{\coNP}{\mathsf{coNP}}
\renewcommand{\P}{\mathsf{P}}
\newcommand{\FP}{\mathsf{FP}}
\newcommand{\GapP}{\mathsf{GapP}}
\newcommand{\BQP}{\mathsf{BQP}}
\newcommand{\QMA}{\mathsf{QMA}}
\newcommand{\VP}{\mathsf{VP}}
\newcommand{\VNP}{\mathsf{VNP}}

\newcommand{\Sch}{\mathcal{S}}

\newcommand{\ot}{\otimes}
\newcommand{\op}{\oplus}

\newcommand{\id}{\mathbbm{1}}

\renewcommand{\phi}{\varphi}
\newcommand{\g}{{r_G}}
\newcommand{\h}{{r_H}}

\def\la{\lambda}
\def\al{\alpha}

\def\ga{\gamma}
\def\vla{\vec{\la}}
\def\vmu{\vec{\mu}}
\def\vnu{\vec{\nu}}

\renewcommand{\Pr}{\mathop{\bf Pr\/}}

\renewcommand{\vec}[1]{\boldsymbol{\mathbf{#1}}}

\newcommand{\irr}[2]{\{#1\}_{#2}}
\newcommand{\restr}{\!\!\downarrow^G_H}

\newcommand\frestr[2]{{\left.\kern-\nulldelimiterspace #1 \vphantom{\big|} \right|_{#2}}}
\allowdisplaybreaks[4]
\begin{document}

\author[1]{Matthias~Christandl}
\author[2]{Aram W.~Harrow}
\author[3]{Greta~Panova}
\author[1]{Pietro~M.~Posta}
\author[4,5,6]{Michael~Walter}
\affil[1]{Department of Mathematical Sciences, University of Copenhagen, Universitetsparken 5, 2100 Copenhagen, Denmark}
\affil[2]{Center for Theoretical Physics -- a Leinweber Institute, MIT, Cambridge, MA, 02139, USA}
\affil[3]{Department of Mathematics, University of Southern California, Los Angeles, CA, 90089, USA}
\affil[4]{Ludwig-Maximilians-Universit\"at M\"unchen, Theresienstr.~37, 80333 M\"unchen, Germany}
\affil[5]{Munich Center for Quantum Science and Technology (MCQST), Schellingstr.~4, 80799~M\"unchen, Germany}\title{Plethysm is in \texorpdfstring{$\#\mathsf{BQP}$}{\#BQP}}
\affil[6]{Korteweg-de Vries Institute for Mathematics \& QuSoft, University of Amsterdam, The Netherlands}
\date{}
\maketitle
\begin{abstract}
Some representation-theoretic multiplicities, such as the Kostka and the Littlewood-Richardson coefficients, admit a combinatorial interpretation that places their computation in the complexity class $\#\mathsf{P}$.
Whether this holds more generally is considered an important open problem in mathematics and computer science, with relevance for geometric complexity theory and quantum information.
Recent work has investigated the quantum complexity of particular multiplicities, such as the Kronecker coefficients and certain special cases of the plethysm coefficients.

Here, we show that a broad class of representation-theoretic multiplicities is in $\#\mathsf{BQP}$.
In particular, our result implies that the plethysm coefficients are in~$\#\mathsf{BQP}$, which was only known in special cases.
It also implies all known results on the quantum complexity of previously studied coefficients as special cases, unifying, simplifying, and extending prior work.
We obtain our result by multiple applications of the Schur transform.
Recent work has improved its dependence on the local dimension, which is crucial for our work.
We further describe a general approach for showing that representation-theoretic multiplicities are in~$\#\mathsf{BQP}$ that captures our approach as well as the approaches of prior work.
We complement the above by showing that the same multiplicities are also naturally in $\mathsf{GapP}$ and obtain polynomial-time classical algorithms when certain parameters are fixed.
\end{abstract}
%\tableofcontents

%=============================================================================
\section{Introduction}\label{sec: intro}
%=============================================================================
Computing multiplicities is a central problem in representation theory and algebraic combinatorics, and is directly relevant to quantum information theory and geometric complexity theory.
In most cases, no closed-form formulas are known.
A natural and fruitful question is to ask whether these multiplicities count combinatorial objects.
If so, then one obtains \emph{positive combinatorial interpretation}.
Stanley stated the problem of finding such interpretations for important representation-theoretic multiplicities in his influential paper~\cite{stanley2000positivity}.
Specific multiplicities of interest include:
the \emph{Littlewood-Richardson (LR) coefficients}, discussed in 1934~\cite{littlewood1934group}, giving the multiplicity of an irreducible representation of the general linear group~$\GL(n)$ in the tensor product of two others;
the \emph{Kronecker coefficients} of the symmetric group, defined in 1938~\cite{Mur38}, giving the multiplicity of an irreducible representation of the symmetric group~$S_n$ in the tensor product of two others;
and the \emph{plethysm coefficients}, first mentioned by Littlewood in 1936~\cite{Littlewood_36} and defined formally in Ref.~\cite{Lit58}, which gives the multiplicity of an irreducible $\GL(n)$-representation in the composition of two irreducible $\GL$-representations (see \cref{sub:our_contributions}).
By Schur-Weyl duality, the Kronecker coefficients can also be defined in terms of the general linear group, and the LR and plethysm coefficients can also be defined in terms of the symmetric group.

Multiplicities arise naturally across a broad range of areas and applications in mathematics, physics, and computer science, and accordingly they have a long history of study using different methods, from representation theory and invariant theory, to geometry and algebraic combinatorics.
For the plethsym coefficients in particular, initial properties were understood via geometric means, notably their stability~\cite{brion1993stable}. The combinatorial investigations started with the unfulfilled hope of finding an interpretation as counting integer points in polytopes (as is known for the LR coefficients, see below), see e.g.\ the works of Kahle and Michalek~\cite{kahle2016plethysm}.
Approaches via tableaux combinatorics were developed by many groups, see e.g.\ Refs.~\cite{de2021plethysms,colmenarejo2024mystery,orellana2022plethysm}.

The general goal is to give an interpretation as counts of some simple-to-describe but exponentially large set.
Ultimately, however, combinatorial formulas are only known in very special situations, see Refs.~\cite{DIP,orellana2024quasi,pps25} for some of the known cases.

%-----------------------------------------------------------------------------
\subsection{The complexity of computing multiplicities}
%-----------------------------------------------------------------------------
While the notion of ``positive combinatorial interpretation'' is not a priori formally defined, the modern interpretation formalizes it as ``proving the computational problem is in~$\#\P$, and deciding positivity is in~$\NP$'', see Refs.~\cite{Pak22,panova2023computational}.
Indeed, the complexity class $\#\P$ captures the problem of counting gadgets (witnesses) of which there might be many but which can be recognized efficiently (in polynomial time).
This perspective uncovered surprising connections between mathematics and computer science and led to new insights in both fields.
The LR coefficients turn out to be in~$\#\P$, as they admit an interpretation as certain tableaux and integer points in polytopes, whereas their positivity is even decidable in~$\P$ \cite{MNS,BurgisserIkenmeyer2013} as a consequence of the Knutson-Tao proof of the saturation conjecture~\cite{KT}.
In contrast, the problems for Kronecker and plethysm coefficients are still wide open~\cite{Pan23}.
We know that both coefficients are $\#\P$-hard to compute and $\NP$-hard to decide positivity of, see Refs.~\cite{IMW17,FI20}.
However, it is not known whether they are in~$\#\P$, that is, whether they admit a positive combinatorial interpretation, or whether their positivity is in~$\NP$.
\footnote{However, it is known that, for a broad class multiplicities, a natural \emph{asymptotic} version of the positivity problem is in $\NP \cap \coNP$~\cite{BCMW17}.}

As representation-theoretic multiplicities are linear-algebraic quantities defined in high dimensions, it is natural to study their complexity from the perspective of quantum computation.
In particular, their very definition as the dimension of the space of multiplicities of some irreducible representation, gives the hope that one could place them in the quantum complexity class~$\#\BQP$, the quantum analogue of~$\#\P$, by the intuition that these multiplicities might count the dimension of easy-to-recognize subspaces in some bigger vector space.
Indeed, in Ref.~\cite{bravyi2024quantum} it was shown that deciding positivity of Kronecker coefficients is in~$\QMA$, the quantum analogue of~$\NP$, and computing a certain multiple of these coefficients is in~$\#\BQP$.
Shortly after, it was shown in Ref.~\cite{ikenmeyer2023remark} that computing the Kronecker coefficients, as well as certain \emph{special} plethysm coefficients (see below), is in $\#\BQP$.
See also Ref.~\cite{christandl2015algcomp}.
In certain parameter regimes, efficient quantum algorithms for computing these multiplicities were proposed in Ref.~\cite{LH24}, however many of those cases were later shown to be polynomial-time computable by classical algorithms as well~\cite{panova2025polynomial}, see also Ref.~\cite{CDW}.
Thus the quantum and classical complexity of representation-theoretic multiplicities is subtle even in special cases.
Here we consider this problem from a general perspective, motivated by the outstanding problem of characterizing the complexity of the plethysm coefficients.

%-----------------------------------------------------------------------------
\subsection{Relevance in algebraic complexity theory and quantum information theory}
%-----------------------------------------------------------------------------
In algebraic complexity theory, geometric complexity theory (GCT) is a program that seeks to establish lower bounds, such as separating~$\VNP$ from~$\VP$, by exploiting the symmetries of the universal polynomials, for example the permanent and determinant of a symbolic matrix.
One approach via algebraic geometry and representation theory goes via so-called \emph{multiplicity obstructions}.
Namely, decompose the coordinate rings of the $\GL$-orbit closures of a $\VNP$-complete polynomial (usually the permanent)  into irreducible components, and compare the multiplicities of these components to the ones corresponding to a universal polynomial for $\VP$ (usually the determinant).
The plethysm coefficients (as well as the Kronecker coefficients) play a natural and significant role in such computations, for example as upper bounds or even as part of explicit formulas for such multiplicities, see Refs.~\cite{BurgisserChristandlIkenmeyer2011, BurgisserIkenmeyer2011,burgisser2017permanent,burgisser2019no,DIP,IP17}.
The GCT program can also be adapted to other situations, such as finding lower bounds for the complexity of the matrix multiplication problem~\cite{BurgisserIkenmeyer2011}, and representation-theoretic multiplicities again turn out to be crucial.

In quantum physics, representation-theoretic multiplicities are ubiquitous since the state of quantum systems are described by (unit) vectors in high-dimensional complex vector spaces, which are acted upon (unitarily) by symmetry groups.
In some situations (e.g., high energy physics) the groups and representations of interest are fixed, in quantum information and computation one naturally considers sequences of groups and representations.
One important setting is the one-body quantum marginal problem.
Indeed, the compatibility of the one-body marginals with an overall pure state is guided by the Kronecker coefficients~\cite{ChristandlHarrowMitchison2007,ChristandlMitchison2006, Klyachko2004QuantumMarginal}, while the plethysm coefficients appear in the fermionic version of the problem~\cite{AltunbulakKlyachko2008}, known as the one-body $N$-representability problem in quantum chemistry~\cite{ColemanYukalov2000}, which has been studied at least since the early 1970s~\cite{Borland1972}.

%-----------------------------------------------------------------------------
\subsection{Summary of contributions}\label{sub:our_contributions}
%-----------------------------------------------------------------------------
We now present our main results, starting with the quantum complexity of plethysm coefficients.
We first give a brief definition of plethysms (see \cref{sec: problem description} for more detail).
Let $V \cong \C^n$ be a finite-dimensional complex vector space.
Recall that the irreducible polynomial complex representations of the general linear group~$\GL(V)$ are labeled by integer partitions with no more than~$n$ parts.
Given such a partition~$\nu$, we denote by $\rho_{\GL(V),\nu} \colon \GL(V) \to \GL(\irr{\nu}{\GL(V)})$ the corresponding irreducible representation, where $\irr{\nu}{\GL(V)}$ denotes the representation space, called \emph{Weyl module}.
If $\nu$ has more than $n$ parts, we define $\rho_{\GL(V),\nu}$ as the zero-dimensional representation.
We abbreviate~$H \coloneqq \GL(V)$, and consider the group $G \coloneqq \GL(\irr{\nu}{H})$ along with one of its polynomial representations~$\rho_{G,\mu}$, where $\mu$ is again an integer partition.
Then the composition
\begin{equation}\label{eq:composition}
    \rho_{G,\mu} \circ \rho_{H,\nu} \colon H \to \GL(\irr{\mu}{G})
\end{equation}
is a representation of~$H = \GL(V)$.
We shall denote the corresponding representation space by~$\{\mu\}\!\!\downarrow^{\GL(\irr{\nu}{H})}_H$ to remember that we regard it as a module of~$H$, by \emph{restriction} along the irreducible representation~$\rho_{H,\nu} \colon H \to G = \GL(\irr{\nu}{H})$.
In general, the representation~\eqref{eq:composition} is reducible and hence it decomposes into irreducible $H$-representations.
Thus:
\begin{align}\label{eq:def_pleth}
    \{\mu\}\!\!\downarrow^{\GL(\irr{\nu}{H})}_H
\;\cong\; \bigoplus_{\la} \underbrace{\{\la\}_H \oplus \dots \oplus \{\la\}_H}_{a^\lambda_{\mu,\nu}}
\;\cong\; \bigoplus_{\la} \{\la\}_H\ot\C^{a^\lambda_{\mu,\nu}}.
\end{align}
The \emph{plethysm coefficient}~$a^\la_{\mu,\nu}$ is defined as the multiplicity of the irreducible module~$\irr{\la}{H}$
in this decomposition.
Here we assume that $\la$ has no more than~$n = \dim(V)$ parts, in which case the coefficient does not depend on the choice of~$n$ but only on the three partitions~$\la, \mu, \nu$.

\begin{prob}[Plethysm coefficients]\label{prob: plethysm}\
\begin{description}
    \item[~~~Input:]
    Three partitions~$\lambda$, $\mu$, $\nu$, encoded as lists of positive integers, in unary.
    \item [~~~Task:]   Compute the plethysm coefficient $a^\la_{\mu,\nu}$.
\end{description}
\end{prob}

\noindent
Encoding the partitions in unary is a standard convention in the literature.
Thus the input size is~$\Theta(\abs\lambda + \abs\mu + \abs\nu)$.
We give two concrete examples of plethysm coefficients:
\begin{itemize}
\item If both $\mu=(m)$ and $\nu=(d)$ are partitions with a single part, then~$\{\mu\}\!\!\downarrow^{\GL(\{\nu\}_H)}_H = \Sym^m(\Sym^d(V))$ can be interpreted as the space of homogeneous polynomials of degree~$m$ in variables that are coefficients on the space of homogeneous polynomials of degree~$d$.
this space is central in geometric complexity theory and is the main object in the Foulkes conjecture~\cite{foulkes1950concomitants}.
\item If $\mu=(m)$ and $\nu=(1^d)=(1,\dots,1)$, then $\{\mu\}\!\!\downarrow^{\GL(\{\nu\}_H)}_H= \Sym^m(\Alt^d(V))$ corresponds to the space of polynomials in the coefficients of a wavefunction of $d$ fermions, each having Hilbert space~$V$.
This is directly relevant to the $N$-representability problem in quantum chemistry~\cite{AltunbulakKlyachko2008, Klyachko2005NRepresentability}.
\end{itemize}
The general case is a major problem in algebraic combinatorics as stated in Ref.~\cite[Problem 9]{stanley2000positivity}.
Prior work~\cite{ikenmeyer2023remark} showed that the former example is in $\#\BQP$, but already the latter example was out of reach for earlier methods.
Here we prove that the general plethysm problem is in $\#\BQP$:

\begin{thm}\label{thm:plethysm}
\cref{prob: plethysm} is in $\#\BQP$.
As a consequence, the problem of deciding if a plethysm coefficient is positive, $a^\lambda_{\mu\nu} > 0$, is in $\QMA$.
\end{thm}

\noindent
The corresponding quantum algorithm is sketched in \cref{fig:plethysm} and we describe its main ideas in \cref{sub:methods}.
The above also implies an analogous result for the \emph{restriction coefficients}~$r_\lambda^\mu$, defined as the multiplicy of an irreducible $S_n$-representation (Specht module) in the restriction of an irreducible~$\GL(n)$-representation (Weyl module) to the subgroup $S_n \subset \GL(n)$ of permutation matrices.
Indeed, the restriction coefficients can be expressed as a positive sum of plethysms: using symmetric functions (see \cref{sec:notation} for definitions), we have $r_\lambda^\mu=\langle s_\mu, s_\lambda[1+h_1+h_2+\cdots]\rangle$.
Thus, computing restriction coefficients is in~$\#\BQP$ and their positivity is in~$\QMA$, as an immediate corollary of \cref{thm:plethysm}.

We obtain \cref{thm:plethysm} by specialization of a general theorem that applies to a broad class of multiplicities known as \emph{branching multiplicities (or coefficients)}.
They are defined as follows:
We consider groups~$G$ and $H$, together with a group homomorphism $\phi\colon H\to G$.
Any representation~$\pi \colon G \to \GL(W)$ of~$G$ then \emph{restricts} to a representation of~$H$ by composition with~$\phi$, i.e., $\pi \circ \phi \colon H \to \GL(W)$.
The corresponding branching multiplicities are defined as the multiplicities of irreducible $H$-representation in the restriction~$\pi \circ \phi$.
Here we consider branching multiplicites in the setting where~$G$ and $H$ are \emph{products of general linear groups}, and where the group homomorphism~$\phi$ and the~$G$-representation~$\pi$ are specified succinctly in terms of representation-theoretic data.
See \cref{prob:branching_computational} for the formal definition.
This encompasses, among others, all coefficients studied in prior work: Kostka numbers, LR coefficients, and Kronecker coefficients, see \cref{sec: representation-theoretic}
(of course, the first two are even in $\#\P \subset \# \BQP$, while for the latter this is an open problem).
It is known that computing branching multiplicities can be done in classical polynomial time when~$G$ and~$H$ are \emph{fixed}.
This holds not just for products of general linear groups, but more generally in the setting of complex reductive (or compact connected) Lie groups~\cite{CDW}.
In contrast, here we allow the groups (as well as the representations) to be part of the input.
This is substantially more difficult since it involves representations that in general have dimensions that are \emph{exponential} in the input size.

\begin{thm}\label{thm: general}
    Computing branching multiplicities for products of general linear groups (\cref{prob:branching_computational}) is in $\#\BQP$.
    As a consequence, deciding their positivity is in $\QMA$.
\end{thm}

\noindent
We prove \cref{thm: general} as \cref{thm: plethysm final} by giving a quantum algorithm that is perfectly sound and complete (if the gate set is such that the circuit can be implemented perfectly), similarly to the algorithms in Refs.~\cite{bravyi2024quantum,ikenmeyer2023remark}.
In particular, this shows that the positivity problem is in~$\QMA_1$ if there exists a universal gate set such that the high-dimensional version of the quantum Schur transform on $(\C^d)^{\otimes N}$ can be implemented exactly for any choice of~$N$ and~$d$.
We discuss the main technical ideas behind our quantum algorithms in \cref{sub:methods}.
Moreover, in \cref{app: A}, we describe a general approach for showing that representation-theoretic multiplicities are in $\#\BQP$ that subsumes the quantum algorithms of this paper as well as those of the prior works~\cite{bravyi2024quantum, ikenmeyer2023remark} and explains their relation.

As discussed, the plethysm coefficients and in fact most representation-theoretic coefficients are $\#\P$-hard, and it is a central open question to determine when they are in $\#\P$, see Refs.~\cite{Pak22,panova2023computational}.
By \cref{thm:plethysm}, a negative answer to this question would imply a breakthrough separation between the complexity classes~$\#\P$ and~$\#\BQP$.
To gain further insight into their computational complexity, we complement our quantum results by showing that the branching multiplicities are also in~$\GapP$, see \cref{sec: classical}, extending earlier results in Refs.~\cite{burgisser2008complexity,PPcomp,pak2024signed}.

\begin{thm}\label{thm:intro gapp}
Computing branching multiplicities for products of general linear groups (\cref{prob:branching_computational}) is in~$\GapP$.
\end{thm}

\noindent
We prove \cref{thm:intro gapp} as \cref{thm:general_gapp}.
We also present a more specialized algorithm for computing the plethysm coefficients, see \cref{thm:pleth_classical}.
In particular, we obtain a polynomial time algorithm when the number of parts of~$\lambda$ and the size of~$\mu$ are fixed, extending the results of Ref.~\cite{panova2025polynomial}.

\subsection{Methods}\label{sub:methods}

\begin{figure}[t]
\centering\scalebox{0.87}{\begin{tikzpicture}
    \node(0,0){\begin{quantikz}[wire
types={n,n,n,n,q,n,n,n,n}]
&\gategroup[wires=9, steps=8, style={dashed, rounded corners, fill=blue!20, fill opacity=.3, inner sep=5pt}]{Embedding} &&&&&\lstick{\footnotesize$\ket{\nu}$}&\setwiretype{q}&\gate[3][3cm]{\mathcal{S}^{-1}}\gateoutput[3]{$|\nu|$}&\setwiretype{n} &\gate[9][4cm]{\mathcal{S}}\gateinput[9]{$|\nu|\cdot|\mu|$}\gategroup[wires=9, steps=2, style={dashed, rounded corners, fill=green!30, fill opacity=0.3, inner sep=5pt}]{Strong Fourier sampling}&
\\
&&\lstick{$\ket{\mu}$}&\setwiretype{q}& \gate[7][2.5cm]{\mathcal{S}^{-1}}\gateoutput[7]{$|\mu|$} &&&&&\qwbundle{}&&\meter[1]{}&
\\
&&&&&&\lstick{\footnotesize$\ket{\psi_\nu}$}&\setwiretype{q}&&\vdots\setwiretype{n}&&\\
\\
\lstick{witness} &&&&&\setwiretype{n}&&&\vdots&\vdots&&\meter[1]{}\setwiretype{q}& \\
\\
&&&&&&\lstick{\footnotesize$\ket{\nu}$}&\setwiretype{q}&\gate[3][3cm]{\mathcal{S}^{-1}}\gateoutput[3]{$|\nu|$} &\vdots\setwiretype{n}&&
\\
&&\lstick{$\ket{\psi_\mu}$}&\setwiretype{q}&&&&&&\qwbundle{}&&&
\\
&&&&&&\lstick{\footnotesize$\ket{\psi_\nu}$}&\setwiretype{q}&&\setwiretype{n}&&
\end{quantikz}};;
\end{tikzpicture}}
    \caption{Quantum circuit for the plethysm coefficient $a^{\la}_{\mu,\nu}$ (\cref{thm:plethysm}).
    The witness is input in the space $\irr{\mu}{\GL(\irr{\nu}{H})}$.
    The first inverse Schur transform embeds it into $\abs\mu$ registers, each isomorphic to $\irr{\nu}{H}$.
    The second layer of inverse Schur transforms embeds each register into~$\abs\nu$ registers of dimension~$n$.
    After the final Schur transform, we measure the registers corresponding to the partition and the Weyl module.
    We accept if and only if the former returns outcome~$\la$ and the latter~$p_0$, where $\ket{p_0}\in\irr{\la}{H}$ is an arbitrary fixed basis state of the Weyl module.
    }\label{fig:plethysm}
\end{figure}
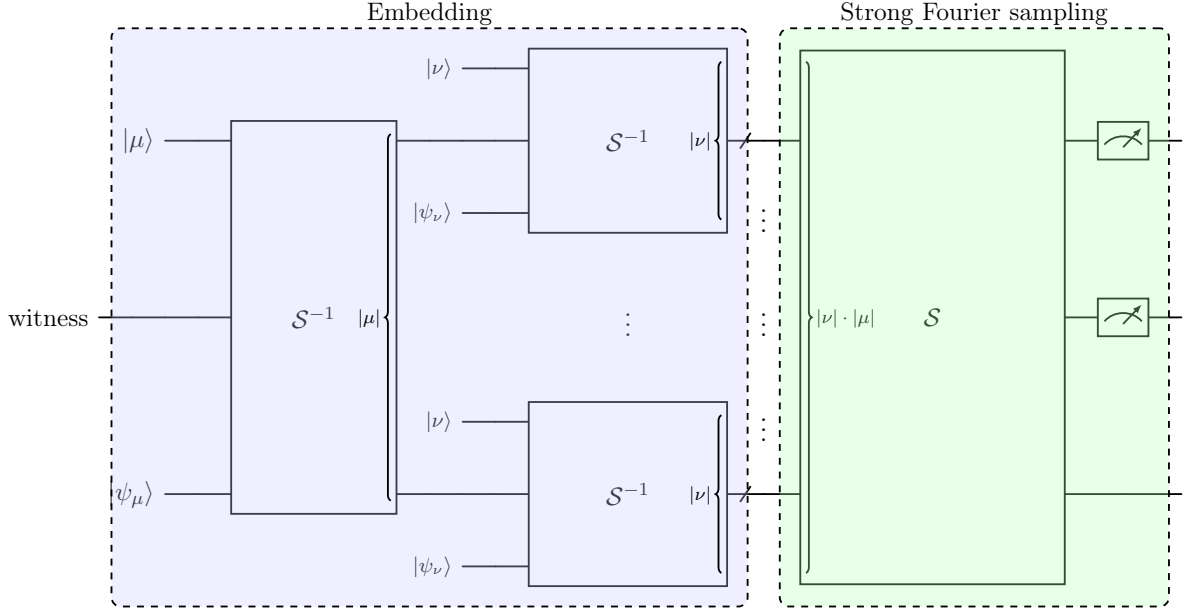

We briefly discuss the main ideas behind our results, focusing on the special case of plethysm coefficients for concreteness (\cref{thm:plethysm}).

\paragraph{Quantum algorithm for plethysm coefficients}
The algorithm in this case is sketched in \cref{fig:plethysm} and formally described in \cref{subsec: plethysm algorithm}.
Let $\la,\mu,\nu$ be partitions and~$V = \C^n$.
We may assume that the number of parts of the partitions and~$n$ are such that the corresponding Weyl modules are nonzero, otherwise the plethysm coefficient~$a^\la_{\mu,\nu}=0$.
Our approach uses Schur-Weyl duality to embed the $H$-module $\{\mu\}\!\!\downarrow^G_H$ inside~$V^{\ot \abs\nu \cdot \abs \mu}$, where $G = \GL(\irr{\nu}{H})$ and $H = \GL(V)$ as before.
This is achieved as follows:
We first embed the $G$-module $\{\mu\}_G$ in a tensor power of $\irr{\nu}{H}$, which is the defining module of~$G = \GL(\irr{\nu}{H})$.
We then consider each $\irr{\nu}{H}$ as a~$H$-module and embed it into a tensor power of $V$, which is the defining module of~$H = \GL(V)$.
Since $\{\mu\}\!\!\downarrow^G_H$ is nothing but the restriction of the $G$-module $\{\mu\}_G$ to~$H$, we obtain the following chain of $H$-equivariant inclusions
\begin{align}\label{eq:pleth emb}
    \{\mu\}\!\!\downarrow^G_H \;\hookrightarrow\; \parens*{ \irr{\nu}{H} }^{\ot \abs\mu} \;\hookrightarrow\; V^{\ot|\nu|\cdot|\mu|}.
\end{align}
In order to implement these embeddings, we employ \emph{inverse} quantum Schur transforms, as these allow us to map from the \emph{Schur basis} of the tensor power spaces to the standard basis (blue box in \cref{fig:plethysm}).
After the embeddings, we implement a final (forward) quantum Schur transform on all~$\abs\nu \cdot \abs\mu$ registers and measure in the Schur basis -- accepting if and only if we obtain the partition~$\la$, meaning that the input state was in the $\la$-isotypic component of $\{\mu\}\restr$, and if we obtain an arbitrary but fixed basis state~$p_0$ on the Weyl module register (green box in \cref{fig:plethysm}).
As inclusions are $H$-equivariant, they preserve the dimension of the isotypic component, and they act as the identity on the Weyl module. As such, the entire circuit implements a projective measurement, and the dimension of the space of witnesses that are accepted with probability one (assuming the quantum Schur transform is implemented exactly) is exactly the multiplicity of~$\irr{\la}{H}$ in~$\{\mu\}\restr$.
By \cref{eq:def_pleth}, this multiplicity is exactly the plethysm coefficient~$a_{\mu,\nu}^\la$.

\paragraph{Generalization to branching multiplicities}
Our result for the branching problem for products of general linear groups (\cref{prob:branching_computational,thm: general}) is based a natural generalization of this idea.
The quantum algorithm is described in \cref{sec: verifier} and analyzed in \cref{thm: plethysm final}.
One conceptual contribution is to identify an appropriate definition of the computational problem.
This is subtle, as it involves objects that will in general be exponentially large and hence need to be succinctly specified in order to obtain a nontrivial computational problem that captures the plethysm coefficients and other classical coefficients as special cases (see \cref{prob:branching_computational,prob:schur_branching}).
On the algorithmic side, we employ a similar sequence of embeddings as before, this time on products of tensor powers of the defining modules, and implemented by a sequence of parallel quantum circuits.
In our general setting, the defining modules for~$G$ need no longer be irreducible as $H$-modules, requiring us to implement intermediate measurements in order to realize the subsequent embeddings.
See \cref{sec: verifier} for further details.

\paragraph{Comparison with prior work}
Previous results on the quantum computation of Kronecker and special plethysm coefficient employed the formulation of these coefficients in terms of the symmetric group~$S_n$ and rely on Beals' quantum Fourier transform for $S_n$~\cite{beals_quantum_1997}.
In principle, any plethysm coefficient can be computed as branching multiplicities of representations of wreath products of symmetric groups.
However, as pointed out in Ref.~\cite{LH24}, a direct use of Beals' quantum Fourier transform leads to an exponential blowup in general.
As such it is not immediately clear how to generalize the symmetric group approach to the general plethysm problem (\cref{prob: plethysm}).
We overcome this challenge by working with the representation theory of the general linear group.
In this regard we may view this approach as the ``Schur-Weyl dual'' of prior approaches to treat representation-theoretic multiplicities.
We emphasize that the local dimensions in our setting (such as the dimension of~$\irr{\nu}{H}$) can be \emph{exponentially large}.
Accordingly, we require a version of the quantum Schur transform that only depends polylogarithmically on the local dimensions.
This has recently been achieved in Ref.~\cite{burchardt_high-dimensional_2025}, see also Ref.~\cite{harrow2005}.
In contrast, for the Kronecker coefficients the relevant local dimension are polynomial and hence the original quantum Schur transform suffices.
We discuss quantum algorithms for the Schur transform in more detail in~\cref{subs: schur transform}.

\paragraph{Unified approach towards multiplicities}
Despite the difference between these approaches there is a unified picture that captures both and which applies even more broadly, to the general problem of the multiplicity of an irreducible $H$-module~$V_\la$ in an~$H$-module $V$, for some group~$H$.
Consider the following two-step procedure:
\begin{enumerate}[noitemsep]
    \item Embed the $H$-module~$V$ we wish to decompose into an $H$-module $V_{\text{model}}$.
    \item Apply strong Fourier sampling for $V_{\text{model}}$ and accept if the outcome consists of is the label of the irreducible representation $V_{\la}$, together with a fixed basis state $\ket{p}\in V_{\la}$.
\end{enumerate}
It is clear that if these two steps can be efficiently implemented, then the multiplicity problem is in~$\#\BQP$, and the corresponding decision problem is in~$\QMA$.
Remarkably, this captures both the approach of this work and the prior works:
\begin{itemize}
\item For the plethysm problem, $V = \{\mu\}\!\!\downarrow^G_H$ and $V_\text{model} = (\C^n)^{\ot \abs\nu \cdot \abs\mu}$.
The first (embedding) step is achieved by the two layers of inverse Schur transforms (blue box in \cref{fig:plethysm}), while the second (strong Fourier sampling) step is achieved by the (forward) Schur transform and subsequent measurements (green box in \cref{fig:plethysm}).
Our more general algorithm for the branching problem of products of general linear groups likewise fits this picture.
\item The key ingredient of Refs.~\cite{bravyi2024quantum,ikenmeyer2023remark} is the \emph{generalized phase estimation} circuit~\cite{harrow2005}, which likewise admits a natural interpretation in this framework.
Here, $V_\text{model} = \C[H] \ot V$, where~$\C[H]$ is the group ring and $H$ acts on $\C[H]$ but not on $V$, the first (embedding) step is realized using the controlled group action of~$V$ (that is, the unitary $\ket h \ket \psi \mapsto \ket h \ket{h \cdot \psi}$ on $\C[H] \ot V$), and the second (strong Fourier sampling) step is implemented using the Fourier transform of the group.
\footnote{In Ref.~\cite{bravyi2024quantum} the authors use \emph{weak} Fourier sampling and therefore prove membership in $\#\BQP$ for a \emph{scaled} version of the Kronecker coefficients (namely, the Kronecker coefficients multiplied by the dimension of corresponding irreducible representations).}
This procedure applies whenever the controlled group action and the Fourier transform of the group can be efficiently implemented, see \cref{cor:gpe bqp}.
\end{itemize}
We can also combine these two ideas.
For example, the restriction coefficients~$r_\lambda^\mu$ can be captured by an embedding step that first uses an inverse Schur transform (as in the first bullet point) to realize~$\irr{\la}{\GL(n)}$ in $V = (\C^n)^{\ot \abs\la}$, following by GPE (as in the second bullet point), noting that the controlled group action of~$S_n$ on~$\C^n$ can be efficiently implemented.
We expand on this general perspective in \cref{app: A}.

\section{Preliminaries and problem description}\label{sec: problem description}
In this section we provide the notation and definitions that we use throughout this manuscript, as well as two definitions of the general branching problem.

\subsection{General notation}\label{sec:notation}

We will use standard notions from combinatorial representation theory and refer the reader to Refs.~\cite{Sag,EC2,JK} for background.

\paragraph{Partitions}

A \emph{partition} of a positive integer $n$ into $d$ parts is denoted $\lambda \vdash_d n$, is a weakly decreasing sequence of $d$ non-negative integers $\lambda=(\lambda_1,\ldots,\lambda_d)$ such that $\lambda_1+\cdots+\lambda_d=n$. We write $|\la| = \la_1+\cdots +\la_d$ for the size of $\la$ and $\ell(\la)$ for the number of nonzero parts (note that $\ell(\la) \leq d)$. A partition is often depicted visually as a \emph{Young diagram}. We note that in our use of partitions and diagrams, we will have to keep track of the zero parts (zero rows) to be consistent with their role as \emph{highest weights}.
The set of partitions of $n$ with at most $d$ parts is denoted by $\mathcal{P}_{d}(n):=\{\la\vdash_d n\}$. The set of all partitions with at most $d$ parts  is $\mathcal{P}_d=\cup_{n>0}\mathcal{P}_d(n)$. We let $\mathcal{P}(n)$
be the set of all partitions of $n$ (disregarding trailing 0s), and $\mathcal{P}=\cup_{n \geq 0} \mathcal{P}(n)$ be the set of all possible partitions.

\paragraph{Irreducible modules}

Let $\GL(d)$ be the general linear group of $d \times d$ complex matrices. Its irreducible modules are the \emph{Weyl modules}, indexed by dominant weights. Throughout the manuscript, we restrict to polynomial representations, in which case the dominant weights can be identified with partitions $\la \in \mathcal{P}_d$. We will denote the irreducible Weyl module indexed by $\la$ as $\{\la\}_d$.

For a tuple $\vec{\la}=(\la^1,\dots,\la^r)\in \mathcal{P}_{d_1}\times\cdots\times \mathcal{P}_{d_r}$, the corresponding irreducible $\GL(d_1)\times\cdots\times\GL(d_r)$-module is denoted $\{\vec{\la}\}_{d_1,\dots,d_r} = \{\la^1\}_{d_1}\ot\cdots\ot\{\la^r\}_{d_r}$, or simply $\irr{\vla}{G}$ if $G:=\GL(d_1)\times\cdots\times\GL(d_r)$.

The irreducible modules of the symmetric group $S_n$ are the \emph{Specht modules}, indexed by $\la \in \mathcal{P}(n)$. We will denote them by $[\la]$. The characters of $S_n$ are denoted by $\chi^\lambda$ and are evaluated at partitions $\alpha\in \mathcal{P}(n)$, which correspond to the conjugacy classes of $S_n$, denoted $C_\alpha$, consisting of all permutations of cycle type $\alpha$.

\paragraph{Symmetric functions and tableaux}

The characters of $\{\la\}_d$ are the \emph{Schur polynomials} $s_\lambda(x_1,\ldots,x_d)$. As $d \to \infty$, these become the Schur functions, which form a basis for the ring of symmetric functions $\Lambda$. Other notable bases for this ring are the homogeneous $\{h_\lambda\}$, elementary $\{e_\lambda\}$, power sum $\{p_\lambda\}$, and monomial symmetric functions $\{m_\lambda\}$, with $\lambda \in \mathcal{P}(n)$ for the degree $n$ component. The Schur functions can be defined as the generating function of \emph{semistandard Young tableaux (SSYT)}, which are assignments of integers in the Young diagram of $\la$ that are weakly increasing along rows and strictly increasing down columns. See Refs.~\cite{EC2,Mac} for background on symmetric functions.

The dimension of the Weyl module $\irr{\la}{d}$ is the number $d_\la$ of semistandard Young tableaux of shape $\la$ with entries in $\{1,\dots,d\}$, is a polynomial in $d$ of degree $|\la|$, and can be efficiently (in $\log(d)$ and $|\la|$) computed via the \emph{hook-content formula}:
\begin{equation}\label{eq: hookcontent}
\dim(\{\la\}_d) = \prod_{(i,j) \in \la} \frac{d+c_\la(i,j)}{h_\la(i,j)},
\end{equation}
where the product is over all boxes $(i,j)$ in the Young diagram of $\lambda$, $c(i,j) = j-i$ is the content of the box and $h(i,j)$ is the hook length of the box. It is upper-bounded by $d^{|\la|}$ and $|\la|^{d(d-1)/2}$.

The dimension of the Specht module $[\la]$ is given by the number $f^\la$ of \emph{standard Young tableaux (SYT)}, which are SSYT of shape $\la$ with each integer $1,\ldots,|\la|$ appearing exactly once, and can be efficiently computed via the \emph{hook length formula}:
\begin{equation}\label{eq: hooklenght}
    f^\la=|\la|!\prod_{(i,j)\in\la}\frac{1}{h_{\la}(i,j)}.
\end{equation}

\subsection{Main representation-theoretic multiplicities}\label{subs:rep theor multipl}
We give equivalent definitions for the main representation-theoretic multiplicities, in terms of multiplicities of irreducible representations as well as in term of decomposition of symmetric functions.
\paragraph{Plethysm coefficients}
The \emph{plethysm coefficients} $a^\lambda_{\mu,\nu}$ were originally defined in Ref.~\cite{Lit58} as the multiplicity of the irreducible $\GL(n)$-module $\{\la\}_n$ in the plethystic composition of $\GL$-modules indexed by $\mu$ and $\nu$, as in~\Cref{eq:def_pleth}, whenever $n\ge\max(\ell(\la),\ell(\nu))$ and $\dim\irr{\nu}{n}\ge \ell(\mu)$.

Understanding decompositions of representations is often done through characters in practice, as the ring of symmetric functions captures all relevant information and can allow us to extract multiplicities as coefficients in polynomial expansions.
In the language of symmetric functions, see Refs.~\cite{EC2,Mac}, the plethysm of two symmetric functions $f$ and $g$, where $g$ has nonnegative integer coefficients and so can be written as a sum of monomials $g=x^{\al^1} + x^{\al^2}+\cdots$, is defined as $f[g]=f(x^{\al^1},x^{\al^2},\cdots)$. The plethysm of Schur functions then decomposes as:
\[
s_\mu[s_\nu] = \sum_\lambda a^\lambda_{\mu,\nu} s_\lambda.
\]
In terms of symmetric group characters, the plethysm coefficients can be computed using the formula:
\[
a^\lambda_{\mu,\nu} = \sum_{\alpha \vdash m}\sum_{\gamma^1,\gamma^2, \ldots \vdash k} \frac{\chi^\mu(\al)}{z_\al} \frac{\chi^{\nu}(\gamma^1)}{z_{\gamma^1}} \frac{\chi^\nu(\gamma^2)}{z_{\gamma^2}} \cdots \chi^\la( \al_1\gamma^1, \al_2\gamma^2,\cdots),
\]
where the sum is over partitions $\al,\gamma^1,\gamma^2,\ldots$, the last character is evaluated at the partition obtained by scaling each partition $\gamma^i$'s parts by $\al_i$ and concatenating those integers to form a partition of $n=mk$. Here $z_\theta = \prod i^{m_i} m_i!$ where $\theta = (1^{m_1},2^{m_2},\ldots)$.

\paragraph{Kronecker coefficients}
The \emph{Kronecker coefficients} $g_{\la,\mu,\nu}$, for $|\la|=|\mu|=|\nu|$, are defined as the multiplicities of the $\GL(d_1)\times\GL(d_2)$ irreducible module $\irr{\nu}{d_1}\ot\irr{\la}{d_2}$ in the $\GL(d_1\cdot d_2)$ module $\irr{\mu}{d_1\cdot d_2}$.
\[
\{\mu\}_{d_1\cdot d_2} \downarrow^{\GL(d_1\cdot d_2)}_{\GL(d_1) \times \GL(d_2)} = \bigoplus_{\substack{\nu, \la}}  \{\nu\}_{d_1} \otimes \{\la\}_{d_2} \ot \C^{g_{\lambda, \mu, \nu}}.
\]
where the restriction is given by the inclusion $\GL(d_1)\times\GL(d_2)\rightarrow\GL(d_1\cdot d_2): (g_1,g_2)\mapsto g_1\ot g_2$.
They can also be defined as the decomposition coefficients for the tensor product of symmetric group representations:
\[
[\mu]\otimes [\nu] = \bigoplus_\la [\la]\ot \C^{g_{\la,\mu,\nu}}.
\]
In terms of symmetric functions, they are defined via the internal product of Schur functions:
\[
s_\la(\mathbf{x}\cdot \mathbf{y}) = \sum_{\mu,\nu}g_{\la,\mu,\nu}s_\mu(\mathbf{x}) s_\nu(\mathbf{y}),
\]
where $\mathbf{x}\cdot \mathbf{y}$ denotes the set of pairwise products of variables $(x_1y_1,x_1y_2,\ldots,x_2y_1,\ldots)$.

\paragraph{Littlewood-Richardson coefficients}

The \emph{Littlewood-Richardson coefficients} $c^\la_{\mu\nu}$ are the decomposition coefficients for the product of Schur functions. They are defined as the multiplicities in the decomposition of the tensor product of two Weyl modules:
\[
\{\mu\}_d\otimes \{\nu\}_d \cong \bigoplus_\la \{\la\}_d\ot\C^{c^\la_{\mu\nu}}.
\]
Equivalently, in terms of Schur functions, they are defined by the product rule:
\[
s_\mu s_\nu = \sum_\la c^\la_{\mu\nu} s_\la.
\]

\paragraph{Kostka numbers}

The \emph{Kostka numbers} $K_{\la\mu}$ are defined as the dimension of the weight space of weight $\mu$ in the $\GL(d)$-module $\{\la\}_d$. Equivalently, they count the number of semistandard Young tableaux of shape $\lambda$ and content $\mu$ (i.e., the number of entries equal to $i$ is $\mu_i$). In the ring of symmetric functions, they serve as the transition coefficients between the Schur, monomial, and homogeneous symmetric function bases:
\[
s_\la = \sum_\mu K_{\la\mu}m_\mu \qquad \text{and} \qquad  h_\mu = \sum_\la K_{\la\mu} s_\la.
\]

\subsection{Branching problem description}

In the following, we describe the computational problem we study in this manuscript. This can be considered as a generalization of the main representation-theoretic multiplicities discussed above.
\paragraph{}
Let $G = \GL(d_1) \times \cdots \times \GL(d_\g)$ and $H = \GL(n_1) \times \cdots \times \GL(n_\h)$ be products of general linear groups, for some positive integers $\g,\h,d_1,\dots,d_\g,n_1,\dots,n_\h$.

Given a group homomorphism $\phi: H \to G$, any representation of $G$ can be viewed as a representation of $H$ by composition, in which case we say that the corresponding $G$-module restricts to $H$. When an irreducible $G$-module $\irr{\vmu}{G}$ is restricted to $H$ we denote it by $\{\vmu\}\restr(\phi)$, or $\{\vmu\}\restr$, omitting the dependency on $\phi$ when the choice is clear. Such a module is generally reducible, decomposing into a direct sum of irreducible $H$-modules:
\begin{equation}\label{eq:restr repeat}
\{\vmu\}\restr \;\cong\;\bigoplus_{\vla}\irr{\vla}{H}\ot \C^{m_{\vmu,\vla}} .
\end{equation}
The sum is over all partition tuples $\vla\in \mathcal{P}_{n_1}\times\cdots\times\mathcal{P}_{n_\h}$ corresponding to irreducible representations of $H$. The non-negative integers $m_{\vmu,\vla}(\phi)$, or simply $m_{\vmu,\vla}$, are called the \emph{branching multiplicities} of $H$ in $G$ (with respect to $\phi$).

Equivalently, the branching multiplicity $m_{\vmu, \vla}$ can be defined as the dimension of the space of $H$-module homomorphisms (\emph{intertwiners}) from the irreducible module~$\irr{\vla}{H}$ to the restricted module~$\{\vmu\}\restr$:
\[
m_{\vmu,\vla}(\phi) := \dim\left(\text{Hom}_H(\irr{\vla}{H}, \{\vmu\}\restr)\right).
\]
We can therefore define the branching problem for products of $\GL$:
\begin{prob}[Branching problem for products of $\GL$]\label{prob:branching_computational}
\
\begin{description}
    \item[~~~Input:]
    \
    \begin{itemize}
        \item The specification of the groups $H = \GL(n_1) \times \dots \times \GL(n_\h)$ and $G=\GL(d_1)\times\cdots\times\GL(d_\g)$
        \footnote{In fact, we will see that the dimensions $d_1,\dots,d_\g$ can be inferred from the specification of the group homomorphism, as such they need not be specified explicitly in this step, and we can assume that they are not given to our quantum algorithm.}
        \item The specification of a group homomorphism $\varphi \colon H \to G$
        \item The specification of irreducible representations $\irr{\vla}{H}$ and $\irr{\vmu}{G}$ for $H$ and $G$ respectively
    \end{itemize}
    \item [~~~Task:]   Compute the branching multiplicity $m_{\vmu,\vla}$, defined as the multiplicity of the irreducible module $\irr{\vla}{H}$ in the restriction $\{\vmu\}\restr$ along the group homomorphism~$\varphi \colon H \to G$~(\cref{eq:restr repeat}).
\end{description}
\end{prob}

To rigorously define the computational problem, we still need to specify the encoding for the input of \cref{prob:branching_computational}. We discuss this in detail in the following.

\subsection{Encoding of the input}\label{subs:encoding}

The branching multiplicities, as defined above, depend on three different kind of inputs, being the definition of the groups $G$ and $H$, the labels of the relative irreducible representations and the specific mapping $\phi:H\rightarrow G$.

The specification of the groups $G = \GL(d_1) \times \cdots \times \GL(d_\g)$ and $H = \GL(n_1) \times \cdots \times \GL(n_\h)$ admits a natural encoding: we simply provide the list of positive integers representing the ranks:
\[
    \g,\h,d_1,\dots,d_\g,n_1,\dots,n_\h,
\]
which are encoded in binary.

The irreducible representations of $H$ and $G$ are labeled by $\h$-tuples and $\g$-tuples of partitions respectively. Our input therefore includes the tuple $\vla = (\la^1, \dots, \la^\h)$ for $H$ and $\vmu = (\mu^1, \dots, \mu^\g)$ for $G$, where:
\[
    \vla \in \mathcal{P}_{n_1}\times\cdots\times\mathcal{P}_{n_\h} \quad \text{and} \quad \vmu \in \mathcal{P}_{d_1}\times\cdots\times\mathcal{P}_{d_\g}.
\]
As is standard for computational problems involving partitions of unbounded length, these labels are provided in \emph{unary} encoding. That is, for a partition $\al$ of size $|\al|=k$, the input size is taken to be $\Theta(k)$. We also need to keep tracks of the zero parts of the partitions, to be consistent with their role as highest weights. This can be done by adjoining each partition with a natural number, in binary, indexing the number of rows. That is, for the partition $\al=(5,4,2)$ labeling an irreducible representation for $\GL(6)$, the encoding we use is $[(00000,0000,00), 6]$.

As for the encoding of the homomorphism $\phi:H\rightarrow G$, the matter is more delicate, mainly because we want \cref{prob:branching_computational} to be a proper generalization of the computation of the standard representation-theoretic multiplicities, and in particular of \cref{prob: plethysm}.

Consider the plethysm case, where $H=\GL(n)$, $G=\GL(\irr{\nu}{H})$ for a irreducible representation of $H$ labeled by $\nu\in\mathcal{P}_n$, and we are interested in representations labeled by $\la\in\mathcal{P}_n$ and $\mu\in\mathcal{P}_{d_\nu}$, where $d_\nu=\dim \irr{\nu}{H}$ is in general exponential in $n+|\nu|$. The input we are given in the corresponding instance of \cref{prob:branching_computational} will be composed of the ranks $n$ and $d_\nu$, the labels $\la$ and $\mu$ and the description of the representation $\phi=\rho_\nu: H\rightarrow G$.

 Whatever encoding of $\phi$ we choose, in order to properly generalize the plethysm problem, we need to be able to recover the label $\nu$ from $\phi$. Moreover we would need the size of the encoding of $\phi$ to be at most polynomial in $\log(n),|\la|,|\mu|$ and $|\nu|$, in order to be able to transfer our complexity results to the plethysm coefficients problem, \cref{prob: plethysm}. In fact since $\phi$ does not depend on $\la$ and $\mu$, we would need the description of $\phi$ to be $\poly(|\nu|,\log(n))$. All standard ways to specify a homomorphism $H\rightarrow G$ (e.g. giving the Jacobian matrix of the map at the identity, $d\phi: \mathfrak{h} \to \mathfrak{g}$) will scale with the dimensions of $H$ and $G$, as such, if we use one of these encodings, the size of the input will, in general, be exponential in the size of our inputs, and will therefore not allow us to generalize, for example, the plethysm problem.

Instead, we will make use of the fact that the branching multiplicities do not depend on the full details of the homomorphism $\phi$ and are in fact fully determined by the decomposition of the defining modules for $G$ as $H$-modules:
\[    \C^{d_i}\restr\cong\bigoplus_{j=1}^{t_i} \{\vnu(i,j)\}_H\ot\C^{m_{ij}}.
\]
This is exactly what we needed to generalize the plethysm case, where the only input required was the single partition $\nu$ defining the full decomposition of the defining module $\irr{\nu}{H}$ for $G$:
\[
\C^{d_\nu}\restr\cong \irr{\nu}{H}.
\]
We make this observation rigorous in the following lemma:
\begin{lem}\label{lem:defining representations}
    Let $H$ be a reductive group, $G= \GL(d_1)\times\cdots\times\GL(d_\g)$ and $\phi=(\phi_1,\dots,\phi_\g),\ \psi=(\psi_1,\dots,\psi_\g)$ be homomorphisms from $H$ to $G$, where $\phi_i,\ \psi_i : H\rightarrow \GL(d_i)$. Suppose for all $i=1,\dots,\g$ that $\phi_i$ and $\psi_i$ are equivalent as representations of $H$, then the branching multiplicities of $H$ in $G$ with respect to $\phi$ are the same as those with respect to $\psi$.
\end{lem}
\begin{proof}
    Suppose we are given two homomorphisms $\phi,\ \psi$ as in the statement of the lemma, we need to prove that for any irreducible $G$-module $\irr{\vmu}{G}$, the restricted $H$-modules are equivalent: $\{\vmu\}\restr(\phi)\cong_H\{\vmu\}\restr(\psi)$.
    For the module $\{\vec{\mu}\}_{G}=\irr{\mu^1}{d_1}\ot\cdots\ot\irr{\mu^\g}{d_\g}$, let $\rho_{\vmu}$ be the corresponding representation. By Schur-Weyl duality (see \cref{subs: schur transform}) applied to each tensor factor separately, we can embed this as a subrepresentation on the $G$-module $(\C^{d_1})^{\ot |\mu^1|}\ot\cdots\ot(\C^{d_\g})^{\ot |\mu^\g|}$. As the components $\phi_i$ and $\psi_i$ of $\phi$ and $\psi$ are equivalent representations, then for each $i=1,\dots,\g$ there exists an intertwiner $f_i\in\GL(d_i)$ such that $\forall x\in H,\ f_i\phi_i(x)=\psi_i(x)f_i$. Then
    \[
    \frestr{f_1^{\ot |\mu^1|}\ot\cdots\ot f_\g^{\ot |\mu^\g|}}{\irr{\vec{\mu}}{G}
    }\in \GL(\irr{\vec{\mu}}{G})
    \]
    is an invertible intertwiner for the $H$-representations $\rho_{\vec{\mu}}\circ \phi$ and $\rho_{\vec{\mu}}\circ \psi$. Thus the two representations are equivalent, and share the same decomposition into irreducible components.
\end{proof}

This allows us to make precise the encoding of the input specifying our group homomorphism:
\begin{dfn}[Specification of group homomorphism]\label{def:specification}
Given $H=\GL(n_1)\times\cdots\times\GL(n_\h)$ and $G=\GL(d_1)\times\cdots\times\GL(d_\g)$, we will specify a group homomorphism by the following data:
\begin{itemize}%[noitemsep]
\item natural numbers~$t_i \in \N$ for $i\in\{1,\dots,\g\}$.
\item natural number~$m_{ij} \in \N$ for $i\in\{1,\dots,\g\}$ and $j\in\{1,\dots,t_i\}$, and
\item partition tuples $\vnu(i,j) \in \mathcal P_{n_1} \times \dots \times \mathcal P_{n_\h}$ for $i\in\{1,\dots,\g\}$ and $j\in\{1,\dots,t_i\}$
\end{itemize}
such that $d_i = \sum_{j=1}^{t_i} m_{ij} \dim \{\vnu(i,j)\}_H$.
This data determines a group homomorphism $\varphi \colon H \to G$, as follows:
For $i=1,\dots,\g$, define the $H$-modules $V_i := \bigoplus_{j=1}^{t_i} \{\vnu(i,j)\}_H \ot \C^{m_{ij}}$. Identify~$V_i \cong \C^{d_i}$, to obtain group homomorphisms $\varphi_i \colon H \to \GL(d_i)$, and define $\varphi := (\varphi_1,\dots,\varphi_\g)$.
\end{dfn}

Of course by \cref{lem:defining representations} the choice of group the homomorphism $\phi$ we use in \cref{def:specification} could be replaced with any other group homomorphism such that the defining modules of $G$ restrict to $H$ in accordance with the given data, but the one we choose has the advantage of identifying the computational basis of the defining modules of $G$ with the \emph{Fourier basis} for the action of $H$ on those same modules. This means that we can freely operate \emph{full} Fourier sampling (i.e. measure partition type, Gelfand-Tsetlin basis and the multiplicity register) for $H$ on any state expressed in the computational basis for $\C^{d_i}$. We will make use of this in the algorithm in \cref{sec: verifier}.

\paragraph{}
Another advantage of defining the branching problem in terms of the decomposition of the defining modules of $G$ is that we can immediately translate the problem in the language of symmetric functions, to define a \emph{dimension-independent} version:

\begin{prob}[Dimension-independent branching problem]\label{prob:schur_branching}
\
\begin{description}
    \item[~~~Input:]
    \
    \begin{itemize}
        \item Natural numbers $\g$ and $\h$.
        \item For each $i \in \{1,\dots,\g\}$, a symmetric function $P_i$, expressed as a linear combination of products of $\h$ Schur functions in independent sets of variables. The data therefore consists of a set of $\h$-tuples of partitions
        \[
        \vec{\nu}(i,j)=(\nu^1(i,j),\dots,\nu^\h(i,j)),\ \text{for }i\in\{1,\dots,\g\},\ j\in\{1,\dots,t_i\},
        \]
        each in $\mathcal{P}^\h$, and corresponding multiplicities $\{m_{i,j}\}_{i,j} \subset \N$. This data defines the symmetric functions:
         \[
        P_i = \sum_{j=1}^{t_i} m_{i,j} s_{\vec{\nu}(i,j)}(\mathbf{x}^1,\dots, \mathbf{x}^\h),
        \]
        where  $\mathbf{x}^1=(x_1^1,x_2^1,\dots),\dots,\mathbf{x}^\h=(x_1^\h,x_2^\h,\dots)$ are independent sets of variables, $s_{\nu^k(i,j)}$ is the Schur function associated to the partition $\nu^k(i,j)$ and
        \[
        s_{\vec{\nu}(i,j)}(\mathbf{x}^1,\dots, \mathbf{x}^\h) := s_{\nu^1(i,j)}(\mathbf{x}^1)\cdots s_{\nu^\h(i,j)}(\mathbf{x}^\h).
        \]
        \item Tuples of partitions $\vmu = (\mu^1, \dots, \mu^\g)\in \mathcal{P}^\g$ and $\vla = (\la^1, \dots, \la^\h)\in \mathcal{P}^\h$ .
    \end{itemize}
    Partitions are given in unary encoding.

    \item [~~~Task:] Compute the coefficient of $s_{\vla}(\mathbf{x}^1,\dots, \mathbf{x}^\h) = s_{\la^1}(\mathbf{x}^1)\cdots s_{\la^\h}(\mathbf{x}^\h)$ in the Schur expansion of the product of plethysms
    \[
        s_{\mu^1}[P_1](\mathbf{x}^1,\dots, \mathbf{x}^\h) \cdots s_{\mu^\g}[P_\g](\mathbf{x}^1,\dots, \mathbf{x}^\h).
    \]
    This coefficient, denoted using the Hall inner product, is the integer
    \[
    \langle s_{\mu^1}[P_1](\mathbf{x}^1,\dots, \mathbf{x}^\h) \cdots s_{\mu^\g}[P_\g](\mathbf{x}^1,\dots, \mathbf{x}^\h), s_{\vec{\la}}(\mathbf{x}^1,\dots, \mathbf{x}^\h) \rangle.
    \]
\end{description}
\end{prob}

\cref{prob:branching_computational} and \cref{prob:schur_branching} are essentially equivalent: in \cref{prob:branching_computational}, the decomposition of a representation in irreducible components is equivalent to the decomposition of its character in terms of irreducible characters, which in this case are products of Schur polynomial in independent sets of variables.
Passing to the ring of symmetric functions $\Lambda$ with an infinite number of variables returns \cref{prob:schur_branching}. The Hall inner product is defined in $\Lambda$, however the resulting identities (expansions of the product of Schur function plethysms into regular Schur functions) hold for any choice of variables. In practice, we can then set some of the variables to 0, reduce to polynomial identities and extract coefficients. Formally, this says that for any sets of input partitions $\mathcal{I}$, there exists a finite number of variables $N_{\mathcal{I}}$ such that the decomposition of $s_{\mu^1}[P_1](\mathbf{x}^1,\dots, \mathbf{x}^\h) \cdots s_{\mu^\g}[P_\g](\mathbf{x}^1,\dots, \mathbf{x}^\h)$ is equivalent to the decomposition of the corresponding constructions over Schur polynomials, as long as the Schur polynomials have $\ge N_{\mathcal{I}}$ nonzero variables. In~\Cref{sec: classical} we show how this is done in practice.

The quantum algorithm we develop in \cref{sec: verifier} will tackle \cref{prob:branching_computational}, while in \cref{sec: classical} we use the formulation of \cref{prob:schur_branching}.

\section{Representation-theoretic multiplicities as branching multiplicities}\label{sec: representation-theoretic}

In this section we present how the general branching multiplicities problem described in \cref{prob:branching_computational} is a natural generalization of the main problems related to representation-theoretic multiplicities. We use the notation developed in \cref{sec: problem description}, and explicitly state when the numbers defining the dimensions of the groups are large enough for \cref{prob:branching_computational} to be equivalent to the dimension-independent \cref{prob:schur_branching}. We represent partitions as Young diagrams.

\paragraph{Plethysm coefficients}
The set of inputs for \cref{prob:branching_computational} to describe the problem of calculating the plethysm coefficient $a_{\mu,\nu}^\la$ is:
\begin{itemize}
    \item The specification of the groups:
    \[ \h=1,\ n_1\in \N \quad \text{and} \quad \g=1,\ d_1=\dim\irr{\nu}{n_1}. \]

    \item The specification of the group homomorphism $\varphi$, as in \cref{def:specification}:
    \[
    t_1=1,\quad m_{1,1}=1, \quad \text{and} \quad \vnu(1,1)=\nu,
    \]
    where $\nu$ is the single partition representing the $H$-module structure of the defining module of $G$.

    \item The specification of the representations:
    \[ \vmu=\mu \quad \text{and} \quad \vla=\la, \]
    where the partition $\mu$ indexes a $G$-module, and the partition $\la$ indexes an $H$-module.
\end{itemize}
Here for \cref{prob:branching_computational} to be dimension-independent we must have $n_1\ge\max(\ell(\la),\ell(\nu))$ and such that $\dim\irr{\nu}{n_1}\ge\ell(\mu)$. As the plethysm coefficient $a_{\mu,\nu}^\la$ is zero whenever $|\mu|\cdot|\nu|\neq|\la|$, we restrict to the case $|\mu|\cdot|\nu|=|\la|$, in which case we can take $n_1=|\la|$ as an upper bound.

\paragraph{Kronecker coefficients}
The set of inputs for \cref{prob:branching_computational} to describe the problem of calculating the Kronecker coefficient $g_{\la^1,\la^2,\mu}$ is:
\begin{itemize}
    \item The specification of the groups:
    \[ \h=2,\ n_1, n_2 \in \N \quad \text{and} \quad \g=1,\ d_1=n_1\cdot n_2. \]

    \item The specification of the group homomorphism $\varphi$, as in \cref{def:specification}:
    \[
    t_1=1,\quad m_{1,1}=1,\quad \text{and} \quad \vnu(1,1)=(\ydiagram{1},\,\ydiagram{1}),
    \]
    where the pair of partitions represents the $H$-module structure on the defining module of $G$.

    \item The specification of the representations:
    \[ \vmu=\mu \quad \text{and} \quad \vla=(\la^1, \la^2), \]
    where the partition $\mu$ indexes a $G$-module, and the pair of partitions $\vla$ indexes an $H$-module.
\end{itemize}
For the problem to be dimension-independent we must have $n_1 \ge \ell(\la^1),\ n_2 \ge \ell(\la^2)$ and $n_1\cdot n_2\ge \ell(\mu)$. A sufficient condition is $n_1,\ n_2\ge |\la^1|=|\la^2|=|\mu|$.

\paragraph{Littlewood-Richardson coefficients}
The set of inputs for \cref{prob:branching_computational} to describe the problem of calculating the Littlewood-Richardson coefficient $c^\la_{\mu, \nu}$ is:
\begin{itemize}
    \item The specification of the groups:
    \[ \h=1,\ n_1 \in \N \quad \text{and} \quad \g=2,\ d_1=n_1,\ d_2=n_1. \]

    \item The specification of the group homomorphism $\varphi$, as in \cref{def:specification}:
    \begin{align*}
        &t_1=1,\quad m_{1,1}=1,\quad \vnu(1,1)=\ydiagram{1}, \\
        &t_2=1,\quad m_{2,1}=1,\quad \vnu(2,1)=\ydiagram{1}.
    \end{align*}
    These correspond to the partitions representing the $H$-module structure on the defining modules of $G$.

    \item The specification of the representations:
    \[ \vmu=(\mu, \nu) \quad \text{and} \quad \vla=\la, \]
    where the pair of partitions $\vmu$ indexes a $G$-module, and the partition $\la$ indexes an $H$-module.
\end{itemize}
As the Littlewood-Richardson coefficient $c^\la_{\mu\nu}$ is nonzero only if $|\mu|+|\nu|=|\la|$, we restrict to this case. For the problem to be dimension-independent we must have $n_1 \ge \max(\ell(\la),\ell(\mu),\ell(\nu))$. A sufficient condition is $n_1\ge|\la|$.

\paragraph{Kostka numbers}
The set of inputs for \cref{prob:branching_computational} to describe the problem of calculating the Kostka number $K_{\la\mu}$ is:
\begin{itemize}
    \item The specification of the groups:
    \[ \h=d,\ n_1=\dots=n_d=1 \quad \text{and} \quad \g=1,\ d_1=d. \]

    \item The specification of the group homomorphism $\varphi$, as in \cref{def:specification}:
    \[
    t_1=d, \quad \text{and for } j \in \{1,\dots,d\}: \quad m_{1,j}=1, \quad \vnu(1,j)= e_j,
    \]
    where $e_j$ corresponds to the tuple of partitions $(0,\dots,\ydiagram{1},\dots,0)$ with the single box at position $j$ (and $0$ is the empty partition). These represent the $H$-module structure of the defining module for $G$.

    \item The specification of the representations:
    \[ \vmu=\la \quad \text{and} \quad \vla=(\mu_1, \dots, \mu_d), \]
    where the partition $\la$ indexes a $G$-module, and the weight $\mu$ indexes an $H$-module.
\end{itemize}
As the Kostka number $K_{\la\mu}$ is zero whenever $|\la|\neq|\mu|$, we restrict to this case. For the problem to be dimension-independent, we must have $d \ge \max(\ell(\la),\ell(\mu))$. A sufficient condition is $d\ge|\la|$.

\paragraph{}

We therefore have an immediate reduction from the standard computation of representation-theoretic multiplicities to \cref{prob:branching_computational}:
\begin{lem}
    \cref{prob:branching_computational}, when specialized on these four families of inputs, is precisely the problem of evaluating the plethysm coefficients, Kronecker coefficients, Littlewood-Richardson coefficients and Kostka numbers for the same inputs.
\end{lem}
\begin{proof}
For each coefficient, compare with the first definition given in \cref{subs:rep theor multipl}.
\end{proof}

\section{Quantum complexity classes and the Schur transform}\label{sec: quantum schur}
In the following section we give a definition of the complexity classes $\QMA$ and $\#\BQP$, and discuss the quantum circuit we will use to implement the Schur transform on our quantum algorithm.

\subsection{Quantum complexity classes}
\begin{dfn}[$\QMA$ -- Quantum Merlin-Arthur]\label{def: QMA}
    A language (with promise) $L \subseteq \{0,1\}^*$ is in the complexity class $\QMA$ if there exists an uniform family of polynomial-time quantum circuits
    $\{V_n\}_{n\in\N}$ and a polynomial $p(n)$ such that for any input string $x$ of length $n$ respecting the promise:
    \begin{align*}
        &\forall x \in L, \quad \exists \ket{\psi}\in (\C^2)^{\ot p(n)} : \Pr[V_n(\ket{x}, \ket{\psi}) \text{ accepts}] \geq \frac{2}{3},
        \\
        &\forall x \notin L, \quad \forall \ket{\psi}\in(\C^2)^{\ot p(n)} : \Pr[V_n(\ket{x}, \ket{\psi}) \text{ accepts}] \leq \frac{1}{3}.
    \end{align*}

\end{dfn}

\begin{dfn}[$\#\BQP$ -- Counting-$\BQP$]\label{def: count BQP}
 A (partial) function $f: \{0,1\}^* \to \mathbbm{N}$ is in $\#\BQP$ if there exists a uniform family of polynomial-size quantum circuits $\{V_n\}_{n\in \N}$ and a polynomial $p(n)$ such that for any input string $x$ of length $n$, the witness space for $V_n$ being $\mathcal{H}=(\C^2)^{\ot p(n)}$ decomposes as
 \[
\mathcal{H}=\mathcal{H}_{\text{acc}}(x)\op \mathcal{H}_{\text{rej}}(x),
 \]
 with $\dim \mathcal{H}_{\text{acc}}(x)= f(x)$ and
 \begin{align*}
     &\forall \ket{\psi}\in \mathcal{H}_{\text{acc}}(x)\quad \Pr[V_n(\ket{x},\ket{\psi})\ \text{accepts}]\ge \frac{2}{3},\\
     &\forall \ket{\psi}\in\mathcal{H}_{\text{rej}}(x)\quad \Pr[V_n(\ket{x},\ket{\psi})\ \text{accepts}]\leq \frac{1}{3}.
 \end{align*}
\end{dfn}
For both \cref{def: QMA} and \cref{def: count BQP}, the constants $2/3$ and $1/3$ can be amplified to be exponentially close to 1 and 0, respectively.

\subsection{Schur transform}\label{subs: schur transform}
Consider the tensor product space $(\mathbbm{C}^d)^{\otimes n}$. On this space, we can define commuting actions of the general linear group $\GL(d)$ and the symmetric group $S_n$, as follows:

The action of $g \in \GL(d)$ is defined by its simultaneous application to each tensor factor:
\[
g \cdot (v_1 \otimes v_2 \otimes \cdots \otimes v_n) = (g v_1) \otimes (g v_2) \otimes \cdots \otimes (g v_n),
\]
for $v_i \in \mathbbm{C}^d$.

The action of $\sigma \in S_n$ is defined by permuting the tensor factors:
\[
\sigma \cdot (v_1 \otimes v_2 \otimes \cdots \otimes v_n) = v_{\sigma^{-1}(1)} \otimes v_{\sigma^{-1}(2)} \otimes \cdots \otimes v_{\sigma^{-1}(n)}.
\]
By \emph{Schur-Weyl duality}, these two actions are maximally commuting, and the space decomposes into a multiplicity-free direct sum of irreducible representations of $\GL(d)\times S_n$:
\[
(\mathbbm{C}^d)^{\otimes n} \cong \bigoplus_{\lambda \in \mathcal{P}_d(n)} \{\lambda\}_d \otimes [\lambda].
\]
Importantly this implies that each irreducible polynomial representation of $\GL(d)$ can be embedded in the tensor product space $(\C^d)^{\ot n}$ as a subrepresentation.
This decomposition motivates a unitary change of basis from the computational basis of $(\mathbbm{C}^d)^{\otimes n}$ to the \emph{Schur basis}, denoted by $\{\ket{\lambda, p, q}\}$. A basis vector $\ket{\lambda, p, q}$ is indexed by a partition $\lambda \in \mathcal{P}_d(n)$, a basis state $p$ from an orthonormal basis for the module $\{\lambda\}_d$, and a basis state $q$ from an orthonormal basis for the module $[\lambda]$. We denote such a transformation $\mathcal{S}$ and we call it the \emph{Schur transform}.

As we could in principle choose different orthonormal bases for $\irr{\la}{d}$ and $[\la]$, there are many possible Schur transforms. In this manuscript we utilize the quantum circuit implementation for the Schur transform from Ref.~\cite{burchardt_high-dimensional_2025}, for which the bases of the Weyl modules and of the Specht modules are the \emph{Gelfand-Tsetlin} basis and the \emph{Young-Yamanouchi} basis, respectively (together commonly referred to as the Gelfand-Tsetlin bases).
Particularly for the Weyl modules, it is useful to discuss the encoding of Gelfand-Tsetlin basis vectors as quantum states. Elements of the Gelfand-Tsetlin basis are indexed by \emph{GT patterns} (or equivalently by semistandard Young tableaux). As the entries of these patterns contain a number in $\{1,\dots,n\}$, they can be represented by $n$-dimensional qudits. Those can then be represented as a subspace of the $\lceil\log_2 n\rceil$-qubit space.
We refer to Ref.~\cite{burchardt_high-dimensional_2025} for the details of the encoding.

The first quantum circuit to implement a Schur transform on a quantum system required gate complexity $\mathcal{O}(nd^4\polylog(n,d))$
\cite{bacon_efficient_2006}
(see also Ref.~\cite{burchardt_high-dimensional_2025}). This would be a problem for our protocol as we need to apply Schur transforms to spaces with exponentially high local dimensions $d$. Recent work shows that it can be done in time $\mathcal{O}(n^4\polylog(n,d))$ \cite{burchardt_high-dimensional_2025}
building on ideas in Refs.~\cite{harrow2005} and \cite{krovi_efficient_2019}.

\begin{thm}[\cite{burchardt_high-dimensional_2025}]\label{thm: high dimensional}
     There exists a quantum circuit that, with error probability $\varepsilon$, implements a Schur transform $\mathcal{S}$ with respect to the Gelfand-Tsetlin bases of the $\GL(d)$ and $S_n$ modules with gate complexity $\mathcal{O}(n^4\polylog(n,d,\varepsilon^{-1}))$.
\end{thm}

\section{Algorithms description and analysis}

In this section we prove the main results of the manuscript.
We start with describing our algorithm in the simpler case of plethysm coefficients, which represents both the main motivating example and, in some sense, the simplest application of our method. We then present the general algorithm for \cref{prob:branching_computational} and prove its correctness. The correctness of the algorithm for the plethysm problem can be inferred as a corollary to the more general proof.

\subsection{Algorithm for \texorpdfstring{\cref{prob: plethysm}}{prob1}}\label{subsec: plethysm algorithm}
Let $n\in \N$, and $\la,\ \mu,\ \nu$ partitions. Let $H=\GL(n),\ G=\GL(d_\nu)$, where $d_\nu:=\dim(\{\nu\}_n)$, we describe a quantum algorithm taking quantum witnesses from the space $\{\mu\}_G$, accepting with high probability elements from a $a_{\mu, \nu}^\la$-dimensional subspace, and rejecting with high probability any element in its orthogonal complement.

\paragraph{Algorithm description}
    The algorithm proceeds in three steps:
    \begin{enumerate}
        \item We embed the witness, representing an element in $\irr{\mu}{d_\nu}$, into $\irr{\nu}{n}^{\ot |\mu|}$ by Schur--Weyl duality, and apply an inverse Schur transform $\Sch^{-1}$.
        \item We embed each register, representing elements in $\irr{\nu}{n}$, into $(\C^n)^{\ot |\nu|}$ by Schur--Weyl duality, to get the global space $((\C^n)^{\otimes|\nu|})^{\ot|\mu|}$ and apply $|\nu|$ local inverse Schur transforms $\Sch^{-1}\ot\dots\ot\Sch^{-1}$ on each register $(\C^n)^{\otimes|\nu|}$.
        \item We apply a Schur transform to the global space $(\C^n)^{\otimes|\nu||\mu|}$, and measure both the partition register and $\GL(n)$-module register, and accept if the outcome is $(\la,p_0)$, for some fixed $\ket{p_0}\in \irr{\la}{n}$.
    \end{enumerate}
Let us look at each step in detail:
\paragraph{1.}
 We add two supplementary registers: the \emph{partition register} and the \emph{symmetric module register}. The partition register will be initialized in state $\ket{\mu}$, while the symmetric-module register will be initialized in a fixed state $\ket{\psi_{\mu}} \in [\mu]$.
We then apply an inverse Schur transform on the 3-registers state, to get a $|\mu|$-register state expressed in the computational basis $\{\ket{i_1\dots i_{|\mu|}}\},\ i_j\in\{1,\dots,d_\nu\}$. By \cref{lem:defining representations}, we can assume that $\GL(n)$ acts on each $i_j$ via some fixed, easy to calculate action, that is we identify the computational basis of $\C^{d_\nu}$ with the Gelfand-Tsetlin basis of $\irr{\nu}{n}$.

\paragraph{2.}
We add a partition register in state $\ket{\nu}$ and a symmetric-module register in state $\ket{\psi_{\nu}}$, for some fixed $\ket{\psi_{\nu}}\in[\nu]$, to each of the $|\la|$ registers. The global state is now expressed in the basis $\{\ket{\nu\  i_1\ \psi_{\nu}}\ot\dots\ot\ket{\nu\ i_{|\la|}\ \psi_{\nu}}\}$. We then apply local inverse Schur transforms on each group of 3 registers, to get an element expressed in the computational basis $\{\ket{i_{1,1}\dots i_{1,|\nu|}}\ot\dots\ot \ket{i_{|\mu|,1}\dots i_{|\mu|,|\nu|}}\}$, where $i_{j,k}\in \{1,\dots,n\}$.

\paragraph{3.}
We relabel the indices of the computational basis as $\{\ket{i_1\dots i_{|\nu||\mu|}}\}$, and a apply a Schur transform on the space $(\C^n)^{\ot |\nu||\mu|}$. The state will now be expressed in the Schur basis $\{\ket{\ga\ p \ q}\}$. We subsequently measure the partition register and $H$-module register. Finally we accept if and only if the outcome of the measurement is $(\la,p_0)$, for some fixed $\ket{p_0}\in \irr{\la}{n}$.
\paragraph{}
A circuit representation of the full algorithm is in \cref{fig:plethysm}

\subsection{Algorithm for \texorpdfstring{\cref{prob:branching_computational}}{prob2}}\label{sec: verifier}
This algorithm can be seen as a generalization of the previous one. The main ideas are the same, but we need to be more careful as the defining modules for $G$ are not guaranteed to be irreducible as $H$-representation, as such we employ an intermediate measurement, which we use to implement the subsequent embedding.
\paragraph{}
For the inputs of \cref{prob:branching_computational} we describe a quantum verifier taking witnesses in the space $\{\vec{\mu}\}_G$, accepting with high probability for a $m_{\vec{\mu},\vec{\la}}$-dimensional subspace and rejecting with high probability from its orthogonal complement.
    \paragraph{Algorithm description}
   The algorithm proceeds in the following steps:
\begin{enumerate}
    \item \label{step:embed}
    We embed the witness state $\ket{\psi}\in \irr{\vmu}{G}$ in the tensor product space $\bigotimes_{i=1}^\g (\C^{d_i})^{\ot|\mu^i|}$. We do this by considering, on each of the $\g$ registers the inclusion given by Schur-Weyl duality:
    \begin{equation}
        \irr{\mu^i}{d_i} \hookrightarrow (\C^{d_i})^{\ot|\mu^i|}.
    \end{equation}
    Implementing this requires adding a partition register (initialized to the partition state $\ket{\mu^i}$) and a permutation register representing any state $\ket{q_0} \in [\mu^i]$.

    \item \label{step:schur_inverse}
    We operate a total of $\g$ inverse Schur transforms, each on the space $(\C^{d_i})^{\ot |\mu^i|}$ for $i=1,\dots,\g$. This expresses the local state in the computational basis of the corresponding $\C^{d_i}$ space.

    \item \label{step:measure}
    By \cref{lem:defining representations}, and as explained in \cref{def:specification}, we can assume the action of $H$ on the computational basis of $\C^{d_i}$ is known. This allows us to use the decomposition:
    \begin{equation}
        \C^{d_i}\restr \cong \bigoplus_{j=1}^{t_i} \irr{\vec{\nu}(i,j)}{H} \ot \C^{m_{i,j}}.
    \end{equation}
    We perform a measurement on each of the $\sum_{i=1}^\g|\mu^i|$ sub-registers to determine the corresponding $H$-isotypic component. After measurement, the state now lives in the space
    \begin{equation}
        \bigotimes_{i=1}^{\g} \bigotimes_{s=1}^{|\mu^i|} \irr{\vec{\nu}(i,j_{i,s})}{H} \ot \C^{m_{i,j_{i,s}}},
    \end{equation}
    where each index $j_{i,s} \in\{1,2,\dots,t_i\}$ is an outcome of a measurement.

    \item \label{step:embed_H}
    Accordingly to the measurement outcomes, we use Schur-Weyl duality to embed each of the $H$-irreducible representation registers into its corresponding tensor product space. This involves adding partition and permutation registers based on the outcomes $\nu^k(i,j_{i,s})$, resulting in the state being in the space:
    \begin{equation}
        \bigotimes_{i=1}^\g\bigotimes_{s=1}^{|\mu^i|}\left(\bigotimes_{k=1}^\h (\C^{n_k})^{\ot|\nu^k(i,j_{i,s})|}\right)\ot \C^{m_{i,j_{i,s}}}.
    \end{equation}
    \item \label{step:second Schur transform} We operate a total of $\h\sum_{i=1}^\g |\mu^i|$ inverse Schur transforms on the spaces $(\C^{n_k})^{\ot|\nu^k(i,j_{i,s})|}$. This expresses the state in the computational basis of the corresponding $\C^{n_k}$ space.
    \item \label{step:rearrange}
    We rearrange the tensor product to group terms by the index $k$. The space can be written as:
    \begin{align}
        & \bigotimes_{k=1}^\h\left(\bigotimes_{i=1}^\g\bigotimes_{s=1}^{|\mu^i|} (\C^{n_k})^{\ot|\nu^k(i,j_{i,s})|}\ot\C^{m_{i,j_{i,s}}}\right) \label{eq:rearranged} \\
        \cong &\bigotimes_{k=1}^\h \left( (\C^{n_k})^{\ot \sum_{i=1}^\g\sum_{s=1}^{|\mu^i|} |\nu^k(i,j_{i,s})|} \ot \C^{\prod_{i=1}^\g\prod_{s=1}^{|\mu^i|} m_{i,j_{i,s}}} \right) \label{eq:simplified} \\
        \cong& \left( \bigotimes_{k=1}^\h (\C^{n_k})^{\ot N_k} \right) \ot \C^{m}, \label{eq:final_form}
    \end{align}
    where $N_k := \sum_{i=1}^\g\sum_{s=1}^{|\mu^i|} |\nu^k(i,j_{i,s})|$ and $m=\prod_{i=1}^\g\prod_{s=1}^{|\mu^i|} m_{i,j_{i,s}}$.

    \item \label{step:final_measure}
    We apply $\h$ Schur transforms on the spaces $(\C^{n_k})^{\ot N_k}$ for $k=1,\dots,\h$, to get $\h$ registers expressed in the Schur basis, and the final register being $H$-invariant. For each of the first $\h$ registers we measure the corresponding partition and $\GL(n_k)$-module subregisters. We accept if and only if the outcome for each register $k$ is $(\la^k, p_k)$ for some predetermined fixed state $\ket{p_k} \in \irr{\la^k}{n_k}$ (we may for example choose $\ket{p_k}$ to be a Gelfand-Tsetlin basis element).
\end{enumerate}
 We show a circuit representation of the algorithm in \cref{fig:genalg}.

\begin{figure}[t]
\centering
\scalebox{0.77}
{
\includegraphics{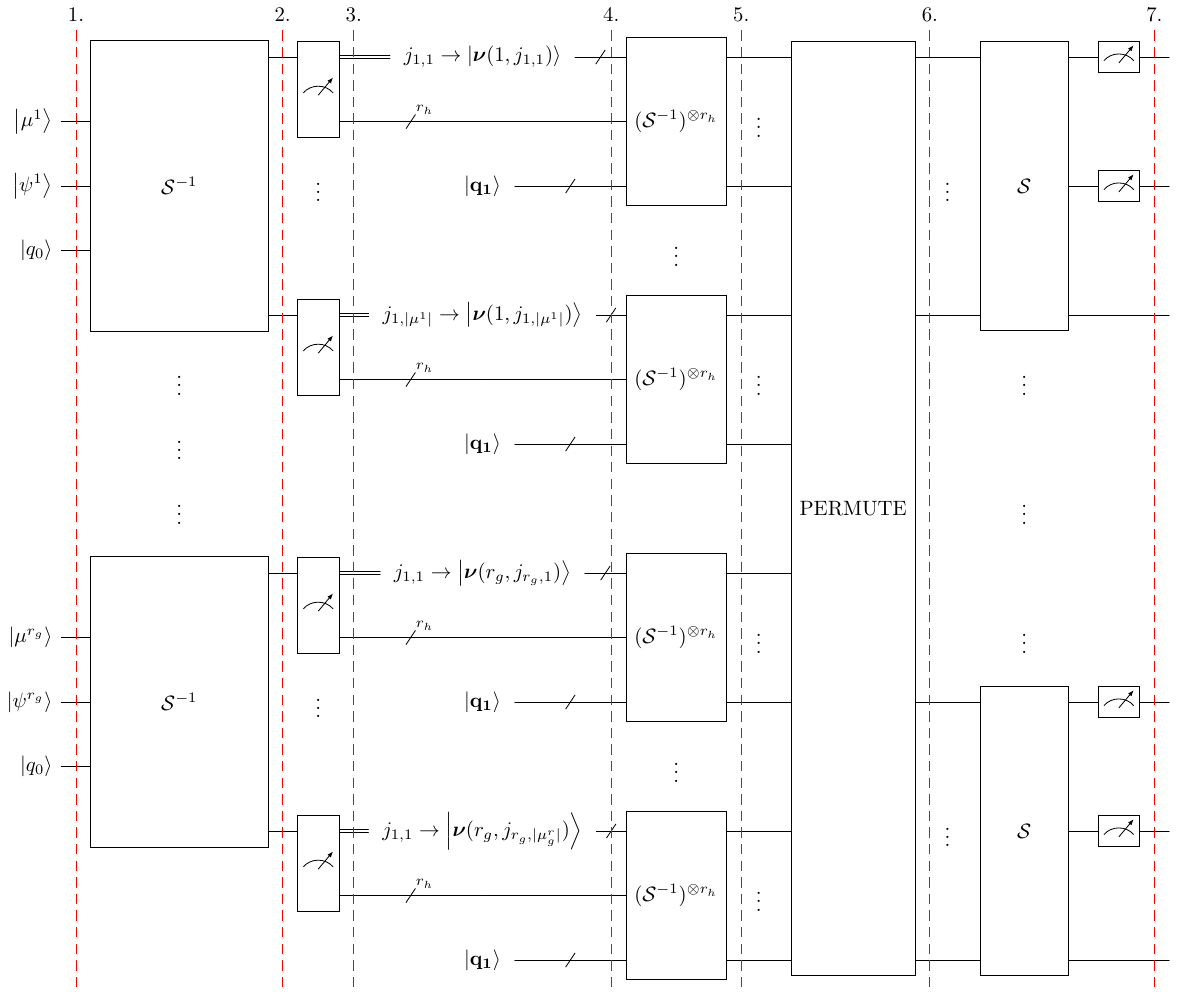}
}
\caption{A circuit representation of the algorithm for \cref{prob:branching_computational}. The states $\ket{q_0}$ and $\ket{\vec{q_1}}$ are arbitrary states used for embedding purposes, and their specific initialization does not impact the algorithm correctness. The permute gate applies the rearrangement described in Step~\ref{step:rearrange}. For clarity purposes we avoid drawing the registers containing the invariant spaces for the $H$-representations coming from the decomposition of the defining modules for $G$ from Step~\ref{step:measure} and onward.}
\label{fig:genalg}
\end{figure}
\paragraph{}
Using this algorithm we can prove:

\begin{thm}[\cref{thm: general}]\label{thm: plethysm final}
    \cref{prob:branching_computational} is in $\#\BQP$ and, consequently, the associated decision problem is in $\QMA$.
\end{thm}
\begin{proof}
We need to prove the completeness and soundness of the algorithm, and that the circuit is polynomial in the input size.

\paragraph{}
As for the completeness and soundness, we claim the algorithm accepts, with a failure probability equal to that of the circuits used to implement the transforms, any witness from the space
\[
W:= \ket{\vec{p}}\ot\C^{m_{\vla,\vmu}}\subseteq\irr{\vla}{H}\ot\C^{m_{\vla,\vmu}}\subseteq\{\vmu\}\restr\quad (\text{see \cref{eq:restr repeat}}),
\]
where $\ket{\vec{p}}:=\ket{p_1}\ot\ket{p_2}\ot\cdots\ot\ket{p_{\h}}$, and rejects, with the same probability, each state from its orthogonal complement $W^\perp\subseteq \{\vmu\}\restr$.
\paragraph{}
Suppose we are given as a witness state $\ket{\psi}\in W$. Since the inclusions of Step~\ref{step:embed} are $H$-equivariants operations, and the Schur transforms of Step~\ref{step:schur_inverse} preserve the $G$-module (and therefore also the $H$-module) structure, our state will now be living in the $\irr{\vla}{H}$-isotypic component of the product space
\[
\bigotimes_{i=1}^{\g}(\C^{d_i})^{\ot |\mu^i|},
\]
and in fact will still be of the form $\ket{\vec{p}}\ot \ket{x}$ for some $\ket{x}$ living in a multiplicity space.

Measuring the label of the $H$-modules in Step~\ref{step:measure} does not change this, as it just tells us which specific product of $H$-modules gives us the copy of $\irr{\vla}{H}$ the state gets projected into. With Steps~\ref{step:embed_H}-\ref{step:final_measure} we perform \emph{strong Fourier sampling}, meaning we measure the irreducible representation label and the basis element for the irreducible module, at the cost of embedding once again the state in a tensor product space. As all operations we executed in the circuit are $H$-equivariant, by Schur's lemma the outcome of the last measurement will be exactly $(\la_k,p_k)$ on each register, and we will therefore accept.

\paragraph{}
Conversely, suppose the witness is in $W^\perp$, by the same reasoning, by Step \ref{step:final_measure} of the algorithm, the state will still be orthogonal to the subspace $\ket{\vec{p}}\ot M$, where $M$ represents the multiplicity space for the $\irr{\vla}{H}$-isotypic component of the module. As such the the measurement is guaranteed (up to the precision of the circuits for the Schur transform) to  reject the measurement result.
\paragraph{}
As for the complexity of the algorithm, in step \ref{step:embed}, the preparation of the auxiliary partition and the permutation register states can be performed efficiently, since we can use the hook length formula to calculate the dimension of the Specht modules $[\mu^i]$ and prepare a state of that size. The $\g$ inverse Schur transforms in Step \ref{step:schur_inverse}, on $\bigotimes_{i=1}^\g (\C^{d_i})^{\ot|\mu_i|}$, have combined gate complexity
\[
\mathcal{O}(\sum_{i=1}^\g|\mu_i|^4\polylog(|\mu^i|,d_i)).
\]
In order to prove that the $d_i$ is at most exponential in the input size (note that this is not needed if we assume $d_i$ is given as part of the input), we write
\[
d_i=\sum_{j=1}^{t_i} {m_{i,j}}\dim(\irr{\vec{\nu}(i,j)}{H}),
\]
and, as $\dim\irr{\nu^k(i,j)}{n}\leq n^{|\nu^k(i,j)|}$), by expanding the definition of $\irr{\vnu(i,j)}{H}$ we get
\begin{align*}
    &d_i=\sum_{j=1}^{t_i} m_{i,j}\prod_{k=1}^{\h} \dim(\irr{\nu^k(i,j)}{n_k})\\
    \leq&\max_{j,k}\ t_i\ m_{i,j}(\dim(\irr{\nu^k(i,j)}{n_k}))^\h\\
    \leq &\max_{j,k}\ t_i\ m_{i,j}\  n_k^{|\nu^k(i,j)|\h}\\
    \in &\,\mathcal{O}(\max_{j,k}\ t_i\ m_{i,j}\exp(|\nu^k(i,j)|\h\log n_k)))
\end{align*}
As a consequence the logarithm of this quantity is polynomial in the input size.
In Step \ref{step:measure}
we perform a polynomial amount of measurements, then, in Step \ref{step:embed_H}, for each of the outcomes we prepare the corresponding partition and symmetric-module register. Again this can be done efficiently. The Schur transforms in Steps~\ref{step:second Schur transform} and \ref{step:final_measure} are transforms on $H$, as such they are polylogarithmic in the $n_k$ and polynomial in the partitions size, and can therefore be efficiently implemented, as there is only a polynomial amount of them. Finally, still in Step~\ref{step:final_measure}, we perform a polynomial amount of measurements.
 This shows that each step can be achieved in polynomial time.
\end{proof}

\section{Classical algorithms via characters}\label{sec: classical}

The problems of computing Kostka, Littlewood-Richardson, Kronecker and plethysm coefficients is central in algebraic combinatorics. While most work is focused on finding positive formulas, efficient signed presentations have also been considered. It is clear from general principles that all of these coefficients are in $\GapP$. Kostka and Littlewood-Richardson coefficients are also equal to the number of integer points in polytopes of dimensions $\mathcal{O}(\ell(\lambda)^2)$, and for them Barvinok's algorithms can be applied~\cite{barvinok1994polynomial}, see Ref.~\cite{panova2025polynomial} for an explicit analysis. The Kronecker coefficients have also been extensively studied, see Refs.~\cite{CDW,PPcomp,burgisser2008complexity}. The computation of plethysm coefficients was considered in some more special cases in Ref.~\cite{FI20}, where it was shown these are in $\GapP$. We will now investigate the general problem.

\begin{thm}\label{thm:general_gapp}
    Computing the dimension-independent branching multiplicities, ~\cref{prob:schur_branching} is in $\GapP$ and can be done in time $\mathcal{O}(N^{\h dn})$, where $N$ is the input size and $d=\sum d_i$, $n=\sum n_i$.
\end{thm}
\begin{proof}
     We are given as input the data $$P_i=\sum_{j=1}^{t_i} m_{i,j} s_{\nu(i,j)}(\mathbf{x}^1,\ldots,\mathbf{x}^\h),$$
    which is a sum of products of Schur functions $s_{\nu^k(i,j)}(\mathbf{x}^k)$. The $P_i$'s are multivariate polynomials in the $\h$-many sets of variables, which we expand into sums of monomials and plug these monomials as variables in $s_{\mu^i}[P_i]$.
A general algorithm would run through the following steps.

\begin{enumerate}
    \item Determine the expansion of $P_i$ into monomials as follows. We have that $$s_{\nu^k(i,j)}(\mathbf{x}^k) = \sum_{\alpha^k(i,j)\vdash n_k} K_{\nu^k(i,j),\alpha^k(i,j)} m_{\alpha^k(i,j)}(\mathbf{x}^k),$$ where $m_{\alpha}$ are the monomial symmetric functions and so
    $$P_i = \sum_{j=1}^{t_i} \sum_{\alpha^k(i,j)\vdash n_k: k=1..h} m_{i,j} \prod K_{\nu^k(i,j),\alpha^k(i,j)} m_{\alpha(i,j)}(\mathbf{x}^1,\ldots,\mathbf{x}^\h), $$
    where $m_{\alpha}(\mathbf{x}^1,\ldots,\mathbf{x}^\h) = m_{\alpha^1}(\mathbf{x}^1)m_{\alpha^2}(\mathbf{x}^2)\cdots m_{\alpha^k}(\mathbf{x}^\h)$.
    The Kostka coefficients can be computed in $\mathcal{O}(n_k^{n_k})$ time (this upper bound can be significantly reduced depending on the partitions $\nu$ involved). The exponential time algorithm can be more efficiently implemented via Barvinok's lattice point count. We can thus compute the coefficients $M_{\alpha(i,j)}$ in time $\mathcal{O}(n^{\h n})$, where $n=n_1+\cdots$ is the total size:
    $$P_i = \sum_{j=1}^{t_i} \sum_{\alpha(i,j)} M_{\alpha(i,j)} m_{\alpha(i,j)}(\mathbf{x}^1,\ldots,\mathbf{x}^\h).$$
    The total number of partitions $\alpha(i,j)$ to do this for is bounded above by $(2^{n})^{n^2}=\mathcal{O}(2^{n^3})$
    \item Compute the expansion of $s_{\mu^i}[P_i]$ into monomials in each set of variables $\mathbf{x}^1,\mathbf{x}^2,\ldots$. This can be done by expanding $s_{\mu^i} = \sum_{\beta^i \vdash d_i} K_{\mu^i,\beta^i} m_{\beta^i}$ in conservatively bounded time $\mathcal{O}(2^{d_i} d_i^{d_i})$. Note that if $f=\sum_{\alpha^i \in R} x^{\alpha^i}$, where $R$ is a multiset of compositions, then
    $$m_\beta[f] = \sum_{\sigma} \sum_{i_1<\cdots <i_r} x^{\sum_j \sigma_j(\beta)\alpha^{i_j}},$$
    where $r =\ell(\beta)$ and $\sigma\in S_r$ goes over all distinct permutations of the parts of $\beta$. The size of $R$ is bounded by $n^{\h n}$ from the previous bounds.
    We can then expand
    $$s_{\mu^i}[P_i] = \sum_{\mathbf{a}} C_{\mathbf{a}}^i (\mathbf{x}^1)^{\mathbf{a}^1} \cdots (\mathbf{x}^\h)^{\mathbf{a}^\h}, $$
    where $\mathbf{a}$ goes over $\h$-tuples of compositions of $|\mu^i|=d_i$, and $C$ are coefficients (which are nonnegative integers). Since the function is symmetric in each set of variables, we have that $C_{\mathbf{a}}$ is invariant under permuting the parts in each $\mathbf{a}^i$. Computing the above expansion can then be bounded by $\mathcal{O}(\prod 2^{d_i}d_i^{d_i} (n^{\h n})^{d_i})=\mathcal{O}(n^{\h d})$ where $d = d_1+d_2+\cdots$.
    \item Take the inner product with $s_{\vla}(\mathbf{x}^1,\ldots,\mathbf{x}^\h)$. To do that we use the fact that $\langle m_{\alpha},h_{\beta}\rangle = \delta_{\alpha,\beta}$, where the $h_\beta$'s are the homogeneous symmetric functions. By the Jacobi-Trudy identity we have
    $$s_\la = \det [h_{\la_i-i+j}]_{i,j=1}^{\ell(\la)} = \sum_{\pi \in S_{\ell(\la)}} {\rm sgn(\pi)} h_{\la-{\rm id} + \pi} ,$$
    where ${\rm id} =(1,2,\ldots,\ell(\la))$ and $\pi$ is treated as a vector likewise. We can then compute
    \begin{align*}
        &\langle s_{\mu^1}[P_1]\cdots s_{\mu^\g}[P_\g], s_{\la^1}(\mathbf{x}^1)\cdots s_{\la^\h}(\mathbf{x}^\h)\rangle \\
        &\qquad = \langle \prod_j  \sum_{\mathbf{a}(j)} C_{\mathbf{a}(j)}^j (\mathbf{x}^1)^{\mathbf{a}(j)^1} \cdots (\mathbf{x}^\h)^{\mathbf{a}(j)^\h} , \prod_i \sum_{\pi^i} {\rm sgn}(\pi^i) h_{\la^i -{\rm id} +\pi^i}(\mathbf{x}^i) \rangle\\
        &\qquad = \sum_{\pi^1,\cdots,\pi^\h \in S_{\ell(\la)}} {\rm sgn}(\pi^1\cdots \pi^\h) \sum_{\substack{ \sum_j \mathbf{a}(j)^i = \la^i-{\rm id}+\pi^i \\ \text{for } i=1,...,\h}} \prod_j C_{\mathbf{a}(j)}^j,
    \end{align*}
    where in the last part we matched the total degree for each set of variables and $\ell(\la)$ is the maximal length of the partitions appearing in $\la$.

\end{enumerate}
    Collecting the positive and negative terms in this formula and realizing that we have exponentially large (but not beyond) sums, puts the last expression in $\GapP=\#\P - \#\P$. The runtime bound comes from the bottleneck of computing the coefficients $C$ in the second step and the sum over all permutations.
\end{proof}

\subsection{Classical computation of plethysms \texorpdfstring{$a^\la_{\mu,\nu}$}{almn}}

The above algorithm is general and has inefficiencies for specific cases. A classical algorithm for plethysms is given in Ref.~\cite[Sec. 9]{FI20}, which clearly shows that the problem is in $\GapP$. However, in some specific instances, one can find a polynomial time algorithm, see e.g.~Refs.~\cite{panova2025polynomial,pps25}. Here we will extend these results.

\begin{thm}\label{thm:pleth_classical}
    Let $\mu \vdash m$, $\nu \vdash d$ and $\la \vdash n:=md$ with $\ell(\lambda)=k$. Suppose that $k$ and $m$ are fixed constants. Then $a^\la_{\mu,\nu}$ can be computed in time $O(d^{m\binom{k}{2}})$ and so the problem of computing the plethysm for such inputs is in $\FP$.
\end{thm}

\begin{proof}
We specialize the analysis from the proof of~\Cref{thm:general_gapp} to the more direct
$$a^\lambda_{\mu\nu}=\langle s_\lambda, s_\mu[s_\nu]\rangle.$$
We have that
\begin{equation}\label{eq:pleth_expansion}
s_\mu[s_\nu(x_1,x_2,\ldots)]= \sum_\lambda a^\lambda_{\mu,\nu} s_\lambda(x_1,x_2,\ldots),
\end{equation}
    which holds for all specializations of the variables. As the Schur functions $s_\lambda$ for $\ell(\lambda) \leq k$ form a basis for the ring of symmetric polynomials in $k$ variables, it is enough to set $\mathbf{x}=(x_1,x_2,\ldots,x_k,0,\ldots)$. We write
    $$s_\nu(x_1,\ldots,x_k) = \sum_{\alpha^i \in A_\nu} \mathbf{x}^{\alpha^i},$$
    where $A_\nu$ is the multiset of weight vectors, i.e. for each $\alpha$ there are $K_{\nu \alpha}$-many vectors (compositions) equal to $\alpha$. The number of terms appearing in that expansion is $D:=s_\nu(1^k)=O(d^{\binom{k}{2}})$. This bound comes from identifying semi-standard young tableaux of shape $\nu$ and entries $\leq k$ with Gelfand-Tsetlin patterns forming a triangle of side $k$ and base $\nu$, which are less than choosing $\binom{k}{2}$-many numbers $\leq \nu_1 \leq d$.

    Next we expand $s_\mu[s_\nu(x_1,\ldots,x_k)] =s_\mu(\mathbf{x}^{\alpha^1},\ldots,\mathbf{x}^{\alpha^D})$. This expansion would have $M$ monomials with $M\leq D^m$, let their degrees be $\beta^1,\ldots,\beta^M$, so that
    $$s_\mu[s_\nu(x_1,\ldots,x_k)] = \sum_{i=1}^M \mathbf{x}^{\sum_{r=1}^D \beta^i_r \alpha^r}.$$
    Now write $s_\lambda(x_1,\ldots,x_k) =\frac{ \det [x_i^{\lambda_j+k-j}]_{i,j=1}^k}{ \det [x_i^{k-j}]_{i,j=1}^k}$ using Weyl's determinantal formula, substitute in~\Cref{eq:pleth_expansion} and multiply both sides by the Vandermonde denominator:
    $$\sum_{\lambda} a^\lambda_{\mu,\nu} \det[x_i^{\lambda_j+k-j}] = \det [x_i^{k-j}] s_\mu[s_\nu(x_1,\ldots,x_k)]= \det [x_i^{k-j}] \sum_{i=1}^M \mathbf{x}^{\sum_{r=1}^D \beta^i_r \alpha^r}$$
    Finally, we can expand the determinant and extract the coefficient $a^\lambda_{\mu,\nu}$ as the coefficient at $x_1^{\lambda_1+k-1}\cdots x_k^{\lambda_k}$ from both sides of this identity and write
    \begin{equation}
        a^\lambda_{\mu,\nu}= \sum_{\sigma \in S_k} {\rm sgn}(\sigma) \sum_{i=1}^M \mathbf{1}[\lambda+\sigma-{\rm id} = \sum_{r=1}^D \beta^i_r\alpha^r]
    \end{equation}
    The total number of computations is then bounded by $k!MD =O(D^m) = O(d^{m \binom{k}{2}})$ and is polynomial when $m,k$ are fixed.
\end{proof}

\begin{rem}
    As shown in Ref.~\cite{FI20} we already know that for $\nu=(3)$, $\mu=(m)$ the problem of computing plethysm is $\#\P$-hard for general $\la$. So we cannot hope for efficient algorithms even in such special cases. There is no known hardness result when $\ell(\lambda)$ is fixed, and we pose this as an open problem.
\end{rem}

\begin{prob}
  Let $\mu\vdash m$, $\nu \vdash d$ and $\lambda \vdash md$ with $\ell(\lambda) =k$ a fixed integer. Does there exist a polynomial (in $md$) time classical algorithm which computes $a^\lambda_{\mu\nu}$?
\end{prob}

 \section*{Acknowledgements}

We are grateful to Christian Ikenmeyer, Vojtech Havlicek and Martin Larocca for fruitful discussions. MC and PP acknowledge financial support from the European Research Council (ERC Grant Agreement No.~818761), VILLUM FONDEN via the QMATH Centre of Excellence (Grant No.~10059) and the Novo Nordisk Foundation (grant NNF20OC0059939 `Quantum for Life'). PP acknowledges support from a PhD scholarship on Quantum Algorithms from the Danish e-infrastructure Consortium (DeiC), and from INdAM-GNFM. GP was partially supported by NSF grant CCF:2302174 and an AMS Birman fellowship, and at residence at ICERM (NSF grant DMS-1929284) in Fall 2025.
MW acknowledges support by the European Union (ERC Grant SYMOPTIC, 101040907) by the Deutsche Forschungsgemeinschaft (DFG, German Research Foundation, 556164098), by the Deutsche Forschungsgemeinschaft (DFG, German Research Foundation) under Germany's Excellence Strategy~--~EXC-2111~--~390814868, and by the German Federal Ministry of Research, Technology and Space (QuSol, 13N17173).
Part of this work was conducted while MW was visiting Q-FARM and the Leinweber Institute for Theoretical Physics at Stanford University and the Simons Institute for the Theory of Computing at UC Berkeley.
AWH was supported by a grant from the Simons Foundation (MP-SIP-00001553, AWH) and supported by a Fulbright Scholar Grant.

%=============================================================================
\appendix
\section{How to prove \texorpdfstring{$\#\BQP$}{\#BQP} for representation-theoretic multiplicities}\label{app: A}
%=============================================================================
In this appendix we outline a general approach towards proving~$\#\BQP$ results for representation-theoretic multiplicities that captures the results of this and prior works.
In this appendix we use the terms module and representation interchangeably.

Let $G$ be a group, and let $V$ be a finite-dimensional $G$-representation that decomposes into irreducible representations.\footnote{For example, this is the case if $G$ is reductive or if $V$ is a unitary representation.}
That is:
\begin{equation}\label{eq: Vdecomposition}
V\cong \bigoplus_{\la\in \Lambda}  \underbrace{V_\la \oplus \dots \oplus V_\la}_{m_\la(V)}\cong \bigoplus_{\la\in\Lambda}V_\la\ot\C^{m_\la(V)},
\end{equation}
where $\Lambda$ labels the isomorphism classes of irreducible finite-dimensional unitary representations of $G$.
We are interested in the problem of computing the multiplicity~$m_\la(V)$ when given as input the group~$G$, the representation~$V$ as well as the label~$\lambda \in \Lambda$ (in some encoding).

\begin{prp}\label{thm: generalapproach}
If, given a reductive group $G$ and a representation $V$ as in \cref{eq: Vdecomposition}, there exists a model representation $V_\text{model}$ for $G$ such that:
\begin{itemize}[noitemsep]
    \item there exists an efficiently-implementable $G$-equivariant isometric embedding $V\hookrightarrow V_{\text{model}}$,
    \item there exists an efficient quantum algorithm to perform strong Fourier sampling on $V_\text{model}$,
\end{itemize}
then the problem of computing the multiplicities $m_\la(V)$ is in $\#\BQP$, and the problem of deciding whether $m_\la(V)>0$ is in $\QMA$ (or in $\QMA_1$ if the above can be implemented exactly).
\end{prp}
\begin{proof}
    The algorithm is given by applying the embedding $V\hookrightarrow V_{\text{model}}$ and subsequently strong Fourier sampling on $V_\text{model}$.
    We accept if the outcome labeling the irreducible representation yields~$\la$, and the outcome relative to the basis states of the irreducible representation~$V_\la$ yields the label~$p$ of an arbitrary fixed basis state~$\ket{p}\in V_\la$.
    A sketch of the circuit is given in \cref{fig: embed+SFS}.
    Note that it implements a projective measurement.
    The space of witnesses that are accepted with probability one is the subspace
    \[
    \ket{p}\ot\C^{m_\la(V)}\subseteq V_\la\ot\C^{m_\la(V)}\subseteq V,
    \]
    which has dimension~$m_\la(V)$.
\end{proof}

\begin{figure}[t]
\centering\scalebox{1}{\begin{quantikz}[wire types={n,q,n}]
&\gate[3][2cm]{\rm Embedding}&\gate[3][5cm]{\rm Strong\ Fourier\ Sampling}&\setwiretype{c}\rstick{$\ket{\la}$}
\\
\lstick{witness}&&&\setwiretype{c}\rstick{$\ket{p}$}
\\
&&&\setwiretype{q}&%\wire[d]{q}\\&&&
\end{quantikz}}
\caption{Quantum circuit to establish that a representation-theoretic multiplicity is in~$\#\BQP$:
Given a witness state in the representation~$V$, first embed it into a model representation $V_\text{model}$.
Subsequently perform strong Fourier sampling, and accept if the outcomes are~$\la$ and the label~$p$ of an arbitrary fixed basis vector~$\ket{p}\in V_\la$. If both steps can be done efficiently, the multiplicity is in~$\#\BQP$.}\label{fig: embed+SFS}
\end{figure}
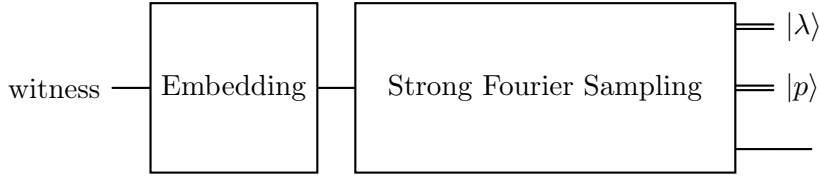

\begin{rem}\label{rem:inv}
We remark that in the notation above
\[
    m_\la(V)=\dim \hom_G(V_\la,V)=\dim \left( V\ot V_\la^* \right)^G,
\]
where $V_\la^*$ is the dual representation to $V_\la$ and~$(\cdot)^G$ is the invariant subspace.
Since the trivial representation has dimension~one, if it is possible to perform the $G$-equivariant embedding $V\ot V_\la^*\hookrightarrow V_\text{model}$ efficiently, then one only needs to apply weak Fourier sampling on~$V_\text{model}$, accepting if the outcome is the label for the trivial representation for $G$.
This is effectively part of the approach Refs.~\cite{christandl2015algcomp,bravyi2024quantum,ikenmeyer2023remark} use for the Kronecker coefficients, where the multiplicity of $[\la]$ in $[\mu]\ot[\nu]$ is interpreted as the dimension of the space of $S_n$-invariants in $[\la]\ot[\mu]\ot[\nu]$ (the representations of~$S_n$ are self-dual).
\end{rem}

We give some examples of this general construction being applied:

\begin{exa}[Plethysms]
    For the plethysm coefficients in~\cref{prob: plethysm}, the embedding step is
    \[
        V \coloneqq \{\mu\}\restr \hookrightarrow V_\text{model} \coloneqq (\C^n)^{\ot |\mu|\cdot|\nu|},
    \]
    and is implemented via two rounds of inverse Schur transforms (see \cref{eq:pleth emb} and discussion below it).
    Strong Fourier sampling for the model representation becomes in this case strong \emph{Schur} sampling, and can be efficiently implemented by a Schur transform followed by measurements.
\end{exa}

\begin{exa}[Branching multiplicities]
    For the branching multiplicities in~\cref{prob:branching_computational}, one can use a construction virtually identical to~\cref{prob: plethysm}, only on multiple representations at the same time.
    The quantum algorithm given is \cref{sec: verifier} is slightly different due to the presence of intermediate measurements.
    This is done only to simplify the algorithm and analysis; they could be made coherent without any loss, resulting in an algorithm of the form of \cref{fig: embed+SFS}.
\end{exa}

\begin{exa}[Kronecker]
    While the Kronecker coefficients~$g_{\lambda,\mu,\nu}$ fall into the general framework of branching multiplicities, they can be directly captured by the above approach.
    Let $\lambda,\mu,\nu$ be partitions of~$n$, each with no more than~$d$ parts.
    Then the Kronecker coefficients are captured by first applying the embedding
    \[
        V \coloneqq [\la]\ot[\mu]\ot[\nu] \hookrightarrow V_\text{model} \coloneqq (\C^d)^{\ot n}\ot(\C^d)^{\ot n}\ot(\C^d)^{\ot n}\cong(\C^{d^3})^{\ot n},
    \]
    using three inverse Schur transforms, and subsequently applying weak Fourier sampling, accepting if the outcome corresponds ot the trivial~$S_n$-representation (as explained in \cref{rem:inv} one could similarly use two inverse Schur transforms and strong Schur sampling).
    Ref.~\cite{christandl2015algcomp} used a close variant of this approach, with the embedding replaced by a projective measurement onto the subspace~$V \subseteq V_\text{model}$.
\end{exa}

%-----------------------------------------------------------------------------
\subsection{Generalized phase estimation}
%-----------------------------------------------------------------------------
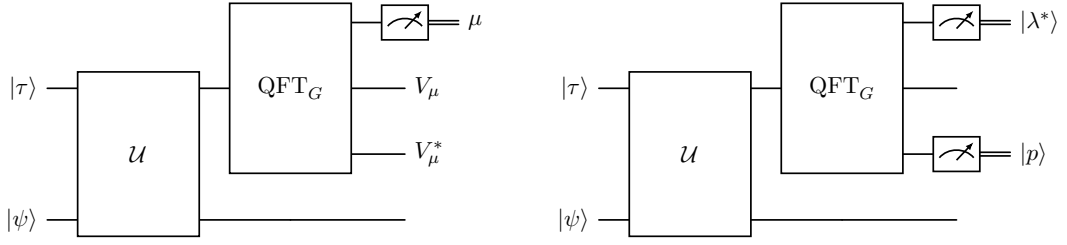
\begin{figure}[t]
    \centering
    \begin{minipage}{0.35\textwidth}
        \centering
        \scalebox{.8}{
            \begin{quantikz}[wire types={n,q,n,q}]
                &&\gate[3][2cm]{\text{${\rm QFT}_G$}}&\meter{}\setwiretype{q}&\setwiretype{c}\rstick{$\mu$} \\
                \lstick{$\ket{\tau}$}&\gate[3][2cm]{\mathcal{U}}&& \rstick{$V_\mu$}\\
                &&&\setwiretype{q} \rstick{$V_\mu^*$}\\
                \lstick{$\ket{\psi}$}&&&
            \end{quantikz}
        }
    \end{minipage}
    \qquad\qquad
    \begin{minipage}{0.35\textwidth}
        \centering
        \scalebox{.8}{
            \begin{quantikz}[wire types={n,q,n,q}]
                &&\gate[3][2cm]{\text{${\rm QFT}_G$}}&\meter{}\setwiretype{q}&\setwiretype{c}\rstick{$\ket{\la^*}$} \\
                \lstick{$\ket{\tau}$}&\gate[3][2cm]{\mathcal{U}}&& \\
                &&&\meter{}\setwiretype{q} &\setwiretype{c}\rstick{$\ket{p}$} \\
                \lstick{$\ket{\psi}$}&&&
            \end{quantikz}
        }
    \end{minipage}
    \caption{Left: Generalized phase estimation (GPE) circuit (with optional uncomputation omitted).
    Right: The circuit underlying \cref{cor:gpe bqp}.
    It is a special case of \cref{fig: embed+SFS}.}
    \label{fig:GPE}
\end{figure}

The approach of Refs.~\cite{bravyi2024quantum,ikenmeyer2023remark} is somewhat more involved and relies on the \emph{generalized phase estimation (GPE)} algorithm~\cite{harrow2005}.
We recall this primitive and explain how it naturally fits \cref{thm: generalapproach}.
Let $G$ a finite group and denote by~$\QFT_G$ the Fourier transform, which identifies
\begin{equation}\label{eq:qft}
    \QFT_G \colon \C[G] \stackrel\cong\longrightarrow \bigoplus_{\mu\in \Lambda} V_\mu\ot V_\mu^*,
\end{equation}
where $\C[G]$ is the group ring, which has a standard basis~$\ket g$ labeled by the group elements~$g \in G$.
This is equivariant for the left and for the right regular action of~$G$ on~$\C[G]$; we denote the latter by~$R \colon \C[G] \to \C[G]$, i.e.\ $R(g)\ket h = \ket{h g^{-1}}$.
Now, consider a unitary representation $\rho \colon G \to \U(V)$ and denote by $\mathcal U$ the controlled group action, i.e., the unitary
\begin{align*}
    \mathcal U \colon \C[G] \ot V \to \C[G] \ot V, \quad \ket g \ot \ket \psi \mapsto \ket g \ot \rho(g)\ket\psi
\end{align*}
Then GPE proceeds as follows~\cite[p.~158]{harrow2005}:
Given as input a state~$\ket{\psi}\in V$,
\begin{enumerate}[noitemsep]
\item prepare a register with Hilbert space~$\C[G]$ in the uniform superposition~$\ket{\tau}:=\frac{1}{\sqrt{|G|}}\sum_{g\,\in G} \ket{g}$,
\footnote{This can be done efficiently by preparing the basis state corresponding to the trivial irreducible representation and applying an inverse Fourier transform.}
\\[-.5cm]
\item apply $\mathcal{U}$,
\item apply the Fourier transform~$\QFT_G$ on the register added in step~1, and
\item measure the label~$\mu$ of the corresponding direct summand in \cref{eq:qft}
\end{enumerate}
optionally followed by an uncomputation (see \cref{fig:GPE}, left).
The first two steps implement an isometry
\begin{align*}
    J \colon V \hookrightarrow V_\text{model} \coloneqq \C[G] \ot V,
    \quad
    \ket\psi \mapsto \frac{1}{\sqrt{|G|}}\sum_{g\,\in G} \ket{g} \ot \rho(g)\ket\psi.
\end{align*}
Importantly, this isometry is $G$-equivariant, with $G$ acting on~$V_\text{model}$ via the right regular representation on~$\C[G]$, that is, $g \mapsto R(g) \ot \id$.
Indeed, for every~$g \in G$, we have
\begin{align*}
    J \rho(g) \ket\psi
= \frac{1}{\sqrt{|G|}}\sum_{h\,\in G} \ket{h} \ot \rho(hg)\ket\psi
= \frac{1}{\sqrt{|G|}}\sum_{h\,\in G} \ket{hg^{-1}} \ot \rho(h)\ket\psi
= (R(g) \ot \id) J \ket\psi.
\end{align*}
Thus we obtain the following as an immediate consequence of \cref{thm: generalapproach}.

\begin{cor}\label{cor:gpe bqp}
If, given a finite group~$G$ and a unitary representation~$V$, there exist efficient quantum circuits for the Fourier transform ${\rm QFT}_G$ and for the controlled group action $\mathcal{U}$, then the problem of computing the multiplicities~$m_\la(V)$ is in $\#\BQP$, and the problem of deciding whether $m_\la(V)>0$ is in $\QMA$ (or in $\QMA_1$ if the above can be implemented exactly).
\end{cor}
\begin{proof}
This follows by reduction to \cref{thm: generalapproach}, with $V_\text{model} \coloneqq \C[G]\ot V$.
By assumption, we can efficiently implement the embedding $J \colon V\hookrightarrow V_\text{model}$ by following the first two steps of GPE, as discussed above.
Strong Fourier sampling for the right regular action on~$V_\text{model}$ can be implemented by applying~$\QFT_G$, measuring the label~$\mu$ of the summand~$V_\mu \ot V_\mu^*$ in \cref{eq:qft} (as in the last two steps of GPE), but outputting~$\mu^*$ and subsequently measuring the register~$V_\mu^*$.
Thus the claim follows from \cref{thm: generalapproach} (see \cref{fig:GPE}, right for the resulting circuit).
\end{proof}

We now comment on the approaches of Refs.~\cite{bravyi2024quantum, ikenmeyer2023remark} and the relation to \cref{cor:gpe bqp}.
To capture the Kronecker coefficients~$g_{\lambda,\mu,\nu}$, Ref.~\cite{bravyi2024quantum} uses weak Fourier sampling on three $\C[S_n]$ registers, accepting if and only if the outcomes are~$\lambda,\mu,\nu$, respectively, which selects the subspace $[\lambda] \ot [\lambda^*] \ot [\mu] \ot [\mu^*] \ot [\nu] \ot [\nu^*]$ of $V = \C[S_n] \ot \C[S_n] \ot \C[S_n]$, as is clear from \cref{eq:qft}.
This is followed by ordinary GPE for the diagonal left regular action of~$S_n$ on~$V$.
The corresponding multiplicites are not the Kronecker coefficients, but rather their multiples~$\dim [\lambda] \dim [\mu] \dim [\nu] \, g_{\lambda,\mu,\nu}$.
Ref.~\cite{ikenmeyer2023remark} observed that a more refined procedure selects a single copy of $[\la]\ot[\mu]\ot[\nu]$ in~$V$, leading to a $\#\BQP$-algorithm for the exact Kronecker coefficients.
To achieve this they use GPE for the Young subgroups~$S_\lambda, S_\mu, S_\nu \subseteq S_n$.
In their setting, this has the same effect as our measurement of the~$V_\lambda^*$ register in the proof of \cref{cor:gpe bqp} (see \cref{fig: embed+SFS}, right).

\begin{exa}[Restriction coefficients]
    Recall that the restriction coefficients~$r_\lambda^\mu$ are defined as the multiplicities of Specht modules in Weyl modules by embedding $S_n \subset \GL(n)$ as permutation matrices.
    That is, they are defined as the multiplicity of the irreducible $S_n$-representation~$[\la]$ in~$\{\mu\}\!\!\downarrow^{\GL(n)}_{S_n}$.
    We can first embed
    \[
        \irr{\mu}{\GL(n)}\hookrightarrow (\C^n)^{\ot |\mu|},
    \]
    by applying an inverse Schur transform (as in the case of plethysms).
    Since the controlled group action of~$S_n$ on $\C^n$ can be implemented efficiently, and there is an efficient QFT for the symmetric group~$S_n$, we can apply \cref{cor:gpe bqp} to see that the restriction coefficients are in~$\#\BQP$.
\end{exa}

{\footnotesize
\bibliographystyle{abbrv}

}

\end{document}